\documentclass[12pt]{article}
\usepackage{psfrag,amsmath,amssymb,latexsym,epic, eepicemu, epsfig,
 float, enumerate}
\usepackage{amscd }
\usepackage{amsfonts}
\usepackage{graphicx}
\usepackage{epsfig,psfrag,amsmath,amssymb,latexsym}
\usepackage{amscd}
\usepackage[dvips]{color}
\usepackage{amsfonts}
\usepackage{graphics}
\usepackage{graphicx}
\usepackage{float}
\usepackage[T1]{fontenc}
\newtheorem{corollary}{Corollary}[section]

\newtheorem{lemma}{Lemma}[section]

\newtheorem{proposition}{Proposition}[section]

\newenvironment{proof}[1][Proof]{\textbf{#1.} }{\ \rule{0.5em}{0.5em}}

\newcommand{\X}{{\cal X}_o}
\newcommand{\e}{\epsilon}

\def\P{{\bf P}}

\numberwithin{equation}{section}

\def\X{{\cal X}_o}
\def\e{\epsilon}
\begin{document}

\title{On the accuracy of the Viterbi alignment}

\author{Kristi Kuljus, Jüri Lember\thanks{Estonian science foundation grant
no 9288 and targeted financing project SF 0180015s12}}

\maketitle

{\small \vspace{0cm} \hbox{\hspace{1.1 cm}\vbox{\noindent\hsize=8cm Swedish University of Agricultural Sciences\\
Department of Forest Economics\\
901 83 Ume{\aa}, Sweden\\
E-mail: Kristi.Kuljus@slu.se}\vbox{\noindent \hsize=8cm  University of Tartu\\
Institute of Mathematical Statistics \\
Liivi 2-513 50409, Tartu,
Estonia\\
E-mail: jyril@ut.ee}} \vskip 1cm \hfill} \vspace{0.5cm}\noindent
%--------------------------------
\abstract{\noindent 
In a hidden Markov model, the underlying Markov chain is usually hidden. Often, 
the maximum likelihood alignment (Viterbi alignment) is used as its estimate.  
Although having the biggest likelihood, the Viterbi alignment can behave very 
untypically by passing states that are at most unexpected. 
To avoid such situations, the Viterbi alignment can be modified by forcing it not to pass these states. 
In this article, an iterative procedure for improving the Viterbi alignment is proposed and studied.
The iterative approach is compared with a simple \textit{bunch approach} where a number of states with
low probability are all replaced at the same time. It can be seen that the iterative way of adjusting
the Viterbi alignment is more efficient and it has several advantages over the bunch approach. The same iterative
algorithm for improving the Viterbi alignment can be used in the case of \textit{peeping}, that is
when it is possible to reveal hidden states. In addition, lower bounds for \textit{classification probabilities} of 
the Viterbi alignment under different conditions on the model parameters are studied. 
}

\vskip 1\baselineskip\noindent {\bf Keywords:} hidden Markov model,
Viterbi alignment, segmentation, classification probability.

\vskip 1\baselineskip\noindent {\bf AMS:} {60J10, 60J22, 62M05}

\section{Introduction and preliminaries}\label{Intro}
\subsection{Notation}
\noindent Let $Y=Y_1,Y_2,\ldots $ be a time-homogeneous Markov chain
with states $S=\{1,\ldots,K\}$ and irreducible  transition matrix
$\mathbb P=\big(p_{ij}\big)$. Let $X=X_1,X_2,\ldots $ be a
process such that: 1) given $\{Y_t\}$  the random
variables $\{X_t\}$ are conditionally independent; 2) the
distribution of $X_j$ depends on $\{Y_t \}$ only through $Y_j$. The
process $X$ is sometimes called a {\it hidden Markov process} (HMP)
and the pair $(Y,X)$ is referred to as a {\it hidden Markov model}
(HMM). The name is motivated by the assumption that the process $Y$,
which is sometimes called the {\it regime}, is non-observable. The
distributions $P_s:=\P(X_1\in \cdot|Y_1=s)$ are called {\it emission
distributions}. We shall assume that the emission distributions are
defined on a measurable space  $({\cal X},{\cal B})$, where ${\cal
X}$ is usually $\mathbb{R}^d$ and ${\cal B}$ is the Borel
$\sigma$-algebra. Without loss of generality we shall assume that
the measures $P_s$ have densities $f_s$ with respect to some
reference measure $\mu$. Our notation differs from the one used in
the HMM-literature, where usually $X$ stands for the regime and $Y$
for the observations. Since our study is mainly motivated by
statistical learning, we would like to be consistent with the
notation used there and keep $X$ for observations and $Y$ for latent
variables. Given a set ${\mathcal A}$ and integers $m$ and $n$, $m<n$, we shall
denote any $(n-m+1)$-dimensional vector with all the components in
${\mathcal A}$ by $a_m^n:=(a_m,\ldots,a_n)$. When $m=1$, it will be
often dropped from the notation and we write $a^n\in {\mathcal A}^n$.
\\\\
%----------------------------------------------------------
HMMs are widely used in various fields of applications, including
speech recognition \cite{rabiner, jelinek}, bioinformatics
\cite{koski, BioHMM2}, language processing  \cite{ochney}, image
analysis \cite{gray2} and many others. For general overview about
HMMs, we refer to \cite{HMMraamat} and \cite{HMP}.
%
%
%----------------
\subsection{Segmentation and standard alignments}
\noindent The {\it segmentation} problem consists of estimating the
unobserved realization of the first $n$ elements of the underlying
Markov chain $Y^n=(Y_1,\ldots, Y_n)$, given the first $n$
observations $x^n=(x_1,\ldots ,x_n)$ from a hidden Markov process
$X^n=(X_1,\ldots,X_n)$ . Formally, we are looking for a mapping
$g:{\cal X}^n \to S^n$ called a {\it classifier}, that maps every
sequence of observations $x^n$  into a state sequence
$g(x^n)=(g_1(x^n),\ldots, g_n(x^n))$, which is often referred to as an {\it
alignment}. Since it is impossible to find the underlying
realization of $Y^n$ exactly, the obtained alignment
$g(x^n)$ has to be the best estimate, in a sense. To measure the goodness of the
obtained alignment (or equivalently of the corresponding classifier),
it is natural to introduce a task-dependent risk function $R(s^n|x^n)$ that gives a measure
of goodness of an alignment $s^n$ given the data $x^n$. For a given risk function, the
best classifier $g$ is then the one that minimizes the risk:
$g(x^n)=\arg\min_{s^n}R(s^n|x^n)$. Such a general risk-based
segmentation theory has been introduced by Lember {\it et al}
\cite{seg, intech} and, independently, by Yau and Holmes \cite{chris}.
The most popular classifier in practice is the
so-called {\it Viterbi  classifier} $v$ that maximizes the posterior
probability, i.e.
$$v(x^n):=\arg\max_{s^n}\P(Y^n=s^n|X^n=x^n).$$
The name  is inherited from the dynamic programming algorithm
(Viterbi algorithm) used for finding it. Obviously, the Viterbi
alignment is not necessarily unique. Despite its popularity, the Viterbi
classifier has some major disadvantages. In particular, the Viterbi
alignment does not minimize the expected number of classification
errors. The best alignment in this sense and therefore also often used in practice
is the so-called {\it pointwise maximum a posteriori (PMAP)} alignment defined as follows:
$$g_t(x^n):=\arg\max_{s\in S} \P(Y_t=s|X^n=x^n), \quad t=1,\ldots,n.$$
Because the value of $g_t(x^n)$ does not depend on $g_{t'}$ for any other $t'\ne t$,
the PMAP-alignment can be obtained pointwise.
Thus, unlike the Viterbi classifier, the PMAP-classifier is purely
local. The lack of global structure is the biggest disadvantage of the
PMAP-classifier, since in the presence of zeros in transition matrix, the alignment
can have zero posterior probability  because of forbidden transitions. Thus, although being best
in the sense of expected number of misclassifications, the PMAP-alignment can have
very low or even zero likelihood. This problem has already been mentioned
in the celebrated tutorial of Rabiner \cite{rabiner} and is probably one of the
main reasons why the Viterbi classifier has become so popular. \\\\
%----------
The Viterbi and PMAP-classifier are both commonly used and can be
considered as the {standard}  classifiers in HMM-segmentation. Both
alignments can be easily found with complexity $O(n)$: the Viterbi
alignment can be found with the Viterbi algorithm and the {\it smoothing
probabilities} $\P(Y_t=s|X^n=x^n)$ (and hence also  the
PMAP-alignment) can be calculated with the well-known forward-backward
recursions. As mentioned above, both of them are, in a sense, extreme. In practice
one would like to have an alignment that has reasonably big
 likelihood (at least non-zero) and at the same time rather small
number of expected classification errors. In \cite{seg, intech},
this goal is aimed at by defining new risk functions, so that the corresponding
best classifiers would in some sense be between the two standard
classifiers and have the properties of both the Viterbi and PMAP. In this paper,
we proceed differently. We take the Viterbi alignment and try to
modify it so, that the expected number of classification errors will decrease, but the
posterior probability of the modified alignment will still remain considerably high. This approach is
motivated by the study of classification probabilities introduced in
the next subsection.
%
%
%-------------------------------------------------------------------------
\subsection{Overview of the main results}
%------------------------
\subsubsection{Bounds for classification probabilities}
Given a classifier $g=(g_1,\ldots,g_n)$, the main object of interest
in this paper is  the probability that for a given time point $t=1,\ldots,
n$, the alignment guesses the true state $Y_t$ correctly:
\begin{equation}\label{classification}
\P(Y_t = g_t(x^n)|X^n=x^n).
\end{equation}
Let us call these probabilities {\it classification probabilities}.
Obviously, this probability tends to decrease when the
number of hidden states $K$ increases, and for any $t$ the
classification probability is biggest when $g$ is the
PMAP-classifier. For the PMAP-classifier (and for any
HMM) the following lower bound trivially holds:
$$\P(Y_t = g_t(x^n)|X^n=x^n)\geq {1\over K},\quad t=1,\ldots, n.$$
Thus, for a two-state HMM one can be sure that given the
observations $x^n$ and a time point $t$, the PMAP-classifier guesses
the hidden state $Y_t$ correctly with probability ${1\over 2}$ at least,
even if the overall  probability of observing the PMAP-state
sequence $g(x^n)$ is very small. Given $x^n$, the sum of the classification
probabilities is just the expected number of correctly classified
states:
$$E\big[\sum_{t=1}^n I_{\{Y_t=g_t(x^n)\}}|X^n=x^n\big]=\sum_{t=1}^n\P(Y_t =
g_t(x^n)|X^n=x^n).$$ In our paper, this expectation is referred to
as the {\it accuracy} of the classifier. The PMAP-classifier is
the most accurate classifier and the trivial lower bound above gives that
its accuracy is at least ${n\over K}$. What about the
Viterbi classifier? Can classification probability 
(\ref{classification}) for the Viterbi classifier be arbitrarily low or
does there exist a data-independent lower bound just like for the PMAP-classifier?
Since all together there are at most $K^n$ different state paths, it follows that the Viterbi path must
have the posterior probability at least $K^{-n}$. Since for any $t$, 
the classification probability is the sum of the posterior probabilities
over all the paths passing $v_t$ at $t$, we obtain the following trivial
lower bound:
\begin{equation}\label{classificationV}
\P(Y_t = v_t(x^n)|X^n=x^n)\geq K^{-n}.
\end{equation}
This bound depends on $n$ and is typically not so useful. Does
there exist a positive lower bound not depending on $n$? These
questions are addressed in Section \ref{sec:clp}. It turns out that
the answer depends on the model. We start with an observation that
when the transition matrix has only positive entries, then a 
data-independent lower bound (that depends on the transition matrix)
exists (Proposition \ref{pos}). Thereafter we present a counterexample
showing that with zeros in the transition matrix this is not
necessarily the case, and for such models classification
probability (\ref{classification}) can be arbitrarily small
(Subsection \ref{counterex}). This counterexample is alarming, since
it shows that although having the biggest likelihood, the Viterbi
alignment can (and when $n$ is big enough, then eventually will)
sometimes behave highly untypically by passing at certain time $t$
a state that is at most unexpected. Hence, for these models there does
not exist a constant data-independent lower bound. However, as shown
in \cite{kuljus}, under some mild conditions there still exists a
data-dependent lower bound (Lemma \ref{lemmauw}). From this lemma it
follows that for a stationary HMM, the tail of the random variable
$$-\ln \P(Y_t = v_t(X^n)|X^n)$$ has an exponential decay independent
of $t$ and $n$. Thus, there exist positive constants $r$ and $d$ so that
for any $t$, any $n$ and any $u>0$,
%\begin{equation}\label{tail}
\[\P\Big(-\ln \P(Y_t = v_t(X^n)|X^n)>u \Big)\leq r \exp[-d u] \]
%\end{equation}
(Corollary \ref{cor2}). Hence, the classification probability can be
arbitrarily small, but such events occur with certain probability
only. As shown in \cite{kuljus}, such a lower bound is useful when 
proving asymptotic results for segmentation.
%-------------
\subsubsection{Modified Viterbi alignment: motivation}
%-------------------------
As explained above, the classification  probabilities of the Viterbi
alignment might be rather small. A small classification probability
at  $t$ means that in most cases the Viterbi alignment guesses the
hidden state $Y_t$ incorrectly. Hence, to control the accuracy, a
natural idea seems to be to modify the Viterbi alignment by
forcing it not to pass such states. More precisely, one
can proceed as follows. Given a threshold parameter $\delta>0$, find all time points
$t$ such that $\P(Y_t= v_t(x^n)|X^n=x^n)\leq \delta$. Let that set be
$T$. Then, for every $t\in T$, find the PMAP-state at $t$, i.e.~find $g_t(x^n)=\arg\max_{s\in S}\P(Y_t=s|X^n=x^n)$.
 After that determine the {\it  restricted Viterbi alignment}
 $$u(x^n):=\arg\max_{s^n\in S^n: s_t=g_t,\,t\in
 T}\P(Y^n=s^n|X^n=x^n).$$
Note that the alignment $u$ equals with the PMAP-alignment at every $t\in T$,
but outside of $T$, the alignment $u$ might still differ from the
Viterbi alignment $v$. In what follows, the described method will be
referred to as the {\it bunch approach}. There are two problems
connected with the bunch approach:
\begin{description}
  \item[1)] Typically a low classification probability entails
 that the Viterbi path has to be isolated for quite a long time. In
 particular, it means that if the classification probability is
 low at some time point $t$, then it is low also in the neighbourhood. Thus,
using the approach above, usually several consecutive  time points
should be replaced by the PMAP-states. This in turn can involve
impossible transitions, so that the obtained alignment $u(x^n)$ can
have zero posterior probability. We shall see in Section
\ref{sec:algorithm} that this can happen.
  \item[2)] Since the alignment $u$ equals with the 
PMAP-alignment at every $t\in T$, the classification probability of $u$ at 
$t\in T$ is biggest possible and (given the threshold $\delta$
is not too big) hence for every $t\in T$,
$\P(Y_t=u_t|X^n=x^n)>\delta$. However, since the alignment $u$ can 
differ from the Viterbi alignment $v$ also outside of $T$, the
 probability $\P(Y_t=u_t|X^n=x^n)$ might
drop below $\delta$ somewhere else.
\end{description}
%-------------------
As a remedy against both mentioned disadvantages, in Section
\ref{sec:algorithm} we propose a more elaborated {\it iterative}
 modification of the Viterbi alignment. To understand the idea of the iterative approach better,
 imagine that at some few time points it is possible to figure out the true
underlying states of hidden $Y$. This can be a realistic
situation in practice, but often figuring out true states 
costs a lot, so this can be done at some few well-chosen time points only.
In what follows, revealing the hidden state shall be called as
{\it peeping} the true state. Since we can not peep often, it is
meaningful to do it at some time point $t$ only if the classification
probability at $t$ is very low, because then the Viterbi alignment
is most likely wrong. Again, one could use the bunch approach: 
figure out the set of time points with lowest misclassification
probabilities and peep them all together, and then find the restricted
Viterbi alignment. Since the revealed states correspond to the true
underlying path, all the transitions in the restricted Viterbi alignment are
possible, and therefore it definitely has positive likelihood. Thus, problem 1) mentioned above disappears.
However, it turns out that peeping is more efficient when it is done iteratively.
\\\\
Start with finding the time point $t_1$ with the lowest classification
probability and peep at $t_1$. Since we now know the value of
$Y_{t_1}$, let it be $y_{t_1}$, we take  this information into
consideration. Thus, in addition to finding the restricted Viterbi
alignment, say $v^{(1)}(x^n)$, it is meaningful to recalculate all the
smoothing probabilities under additional condition
$Y_{t_1}=y_{t_1}$. Hence, we find the conditional classification
probabilities
$$\P(Y_t=v^{(1)}_t(x^n)|X^n=x^n,Y_{t_1}=y_{t_1}),\quad t=1,\ldots,n.$$
Next, find the time point $t_2$ with the smallest conditional classification
probability, peep at $t_2$ and determine the restricted Viterbi alignment,
i.e.~the maximum likelihood path that passes $y_{t_1}$ at $t_1$ and
$y_{t_2}$ at $t_2$. Then calculate again the conditional classification
probabilities by conditioning on $Y_{t_1}=y_{t_1}$ and $Y_{t_2}=y_{t_2}$. Thereafter, find $t_3$ with the lowest
conditional classification probability and so on. In Section
\ref{sec:algorithm}, we present simulations that demonstrate that the
iterative approach is more efficient than the bunch approach, because the
same effect, that is a certain decrease in the number of  classification errors, can be
obtained with considerably fewer number of peepings. 
%-------
When  peeping is not possible, then instead of the true state we consider
the PMAP-state as the one being the most likely hidden state. Thus, in this case  
the iterative algorithm uses PMAP-replacements. Again,
the simulation examples in Section \ref{sec:algorithm} demonstrate the
advantage of the iterative algorithm over the bunch approach also in the case of PMAP-replacements. The
explicit description of the iterative algorithm is given in Section
\ref{sec:algorithm}.
%--------------
\subsubsection{Unsuccessful peeping}
It turns out that the question of choosing the right peeping points
is more important as it might seem at first sight. Indeed, if we
peep at time point $t$ and see that the Viterbi alignment guesses
$Y_t$ correctly, i.e.~$v_t=Y_t$, then  the restricted Viterbi
alignment coincides with the original one and hence, nothing
changes. On the other hand, if $Y_t\ne v_t$, then the restricted
Viterbi alignment differs from the original one and typically more than
just at $t$. Is the average number of correctly classified states
now bigger? Clearly, peeping induces one correctly estimated state,
because $Y_t$ is correct. However, in Section \ref{sec:peep}
we present a counterexample illustrating that it is possible that the
restricted Viterbi alignment behaves so badly in the neighbourhood of
$t$, that the accuracy (the average number of correctly classified
states) drops significantly. In other words, despite the fact that the
restricted Viterbi alignment guesses one more state correctly,
the  average number of correctly classified states for the restricted alignment is worse
than for the unrestricted Viterbi. Therefore, in this example
peeping either does not change anything or makes the alignment even
worse, so that the overall effect of peeping is negative! Moreover,
we show that the example can be constructed so that the 
expected number of classification errors induced by peeping can
be arbitrarily large.\\\\
%-------------
In this example, the badly chosen $t$ has high classification
probability. Thus, peeping at such $t$ does not make much sense, and
neither the bunch nor the iterative approach would pick $t$ as a possible
peeping time. However, it is intriguing to know whether it would be
possible to have such counterexamples also with lower classification
probabilities. More generally, would it be possible to find out (based
on the data $x^n$ and the model) whether the effect of peeping at
$t$ is non-negative? And are there any models (two-state HMMs or 
HMMs with positive transitions, perhaps), where peeping is
guaranteed to have a non-negative effect only? These questions are
the subject of the future research.\\\\
%-------------------
The paper is organized as follows. In Section \ref{sec:clp}, the
classification probabilities and their lower bounds are studied.
Section \ref{sec:algorithm} is devoted to the
iterative algorithm and to simulations illustrating its
behavior. Section \ref{sec:peep} presents the counterexample showing that
peeping can increase the expected number of classification
errors.
%
%
%
%---
\section{Lower bounds for classification
probabilities}\label{sec:clp}
%---------------------------
In this section, we study classification probabilities
(\ref{classification}) for the Viterbi alignment. Recall that the
accuracy of an alignment is just the sum of the corresponding
classification probabilities. At first we note that when all the
transition probabilities are positive, then there exists a data-independent lower
bound to the classification probabilities, hence there exists also a
lower bound to the accuracy that is linear in $n$. Then we present
a counterexample showing that in the presence of forbidden
transitions this is not the case, and the classification probability
can be arbitrarily low. Finally, we prove that under an additional
condition, low classification probabilities
% (as function of observations)
occur with certain small probability only.\\\\
%--------------
Throughout the paper we shall use the following notation. 
For any sequence of observations $x^n$ and any state
sequence $y^n$, $p(x^n)$ stands for the likelihood and $p(x^n,y^n)$ for  the joint likelihood.
For any $s\in S$ and $k=1,\ldots,n$, define the $\alpha$-variables
\begin{align*}
\alpha(x^k,s)&:=\sum_{y^k:y_k=s}p(x^k,y^k),\quad
\alpha(s,x_k^n):=\sum_{y_k^n:y_k=s}p(x_k^n,y_k^n).\quad
%p(x^n):=\sum_s\alpha (x^n,s)
\end{align*}
Thus
\[p(x^n)=\sum_s\alpha (x^n,s).\]
Finally, let for any $s\in S$ and $t=1,\ldots,n$, 
$$\gamma_t(s):=\P(Y_t=s|X^n=x^n)p(x^n).$$
When the emission distributions are discrete, then
\begin{align*}
\alpha(x^k,s)&=\P(X^k=x^k,Y_k=s),\quad
\alpha(s,x^k)=\P(X_k^n=x_k^n,Y_k=s),\\
 p(x^n)&=\P(X^n=x^n),\quad
\gamma_t(s)=\P(X^n=x^n,Y_t=s).\end{align*}
%---------------------
\subsection{Positive transitions}
We need some additional notation. Recall that $\mathbb P=(p_{ij})$
denotes the transiton matrix of $Y$. Let
\begin{equation}\label{sigma}
\sigma_1=:\min_s{\min_{s'}p_{ss'}\over \max_{s'}p_{ss'}},\quad
\sigma_2=:\min_s{\min_{s'}p_{s's}\over \max_{s'}p_{s's}}.
\end{equation}
Clearly, $\sigma_1>0$ if and only if all the transitions are positive
and the same holds for $\sigma_2$.
%-----------------
The following proposition is a special case of Proposition 4.1 in \cite{kuljus}.
The proof is given in Appendix.
\begin{proposition}\label{pos} Assume that all the transition probabilities are positive.
Let $\pi$ be arbitrary initial distribution with $K_1$ non-zero entries. Then
the following bounds hold:
\begin{align*}
\P(Y_t=v_t(x^n)|X^n=x^n)&\geq {\sigma_1^2\sigma_2^2 \over
\sigma_1^2\sigma_2^2 +(K-1)},\quad t=2,\ldots,n-1,\\
\P(Y_1=v_1(x^n)|X^n=x^n)&\geq {\sigma_1^2 \over \sigma_1^2
+(K_1-1)},\quad \P(Y_n=v_n(x^n)|X^n=x^n)\geq {\sigma_2^2 \over
\sigma_2^2 +(K-1)}.\end{align*}
\end{proposition}
%-----------------
Note that for $t=2,\ldots,n$, the lower bounds for the classification probabilities do
not depend on the initial distribution. Hence, the bounds hold also
for stationary distribution. For a stationary chain, the Viterbi
alignment as well as the smoothing probabilities do not depend on
whether the forward or backward chain is considered. Hence, for the time-reversed chain,
the bounds should remain the same, provided that $\sigma_1$ and $\sigma_2$
correspond to the time-reversed chain. Let $\pi$ be now
the stationary distribution and let $q_{ss'}$ denote the transition
probabilities for the reversed chain, then:
$$q_{ss'}={p_{s's}\pi_{s'}\over \pi_{s}}.$$
Let $\sigma'_1$ and $\sigma'_2$ be the minimum values as in (\ref{sigma}) corresponding
to the reversed chain. If the underlying Markov chain is reversible, 
then $q_{ss'}=p_{ss'}$ and $\sigma'_i=\sigma_i$, $i=1,2$. When $\pi$ is
uniform, then $q_{ss'}=p_{s's}$ ($\mathbb P$ is double-stochastic random matrix), hence
$\sigma'_1=\sigma_2$ and $\sigma'_2=\sigma_1$. In both cases $\sigma'_1\sigma'_2 = \sigma_1\sigma_2$
and the lower bounds for $t=2,\ldots, n-1$ remain unchanged. But in general,
$\sigma'_1\sigma'_2\ne \sigma_1\sigma_2$, thus the following corollary is meaningful.
%
%----------------
\begin{corollary} Assume that all the transition probabilities are positive. Then, if 
the initial distribution is stationary, the following bounds hold:
%\begin{equation}\label{double}
\[
\P(Y_t=v_t(x^n)|X^n=x^n)\geq {(\sigma_1\sigma_2\vee
\sigma'_1\sigma'_2)^2 \over (\sigma_1\sigma_2\vee
\sigma'_1\sigma'_2)^2 +(K-1)},\quad t=2,\ldots,n-1,\]
$$\P(Y_1=v_1(x^n)|X^n=x^n)\geq {(\sigma_1\vee \sigma'_2)^2 \over
(\sigma_1\vee \sigma'_2)^2   +(K-1)},$$
$$\P(Y_n=v_n(x^n)|X^n=x^n)\geq {(\sigma_2\vee \sigma'_1)^2 \over
(\sigma_2\vee \sigma'_1)^2   +(K-1)}.$$
\end{corollary}
\begin{proof} The proof follows from the fact that when $a>b>0$,
then
$${a\over a+(K-1)}>{b\over b+(K-1)}.$$
\end{proof}
%---------------------------------
\\\\
%-------------------------------------------
{\bf Example}. An important two-state HMM is the model with
transition matrix
$$\mathbb P=\left(
  \begin{array}{cc}
    1-\epsilon_1 & \epsilon_1 \\
    \epsilon_2 & 1-\epsilon_2 \\
  \end{array}
\right),$$ where $0<\epsilon_1,\epsilon_2\leq 0.5$. Without loss of
generality, let $\epsilon_1\leq \epsilon_2$. Then
$\sigma_1={\epsilon_1\over 1-\epsilon_1}$ and
 $\sigma_2={\epsilon_1\over 1-\epsilon_2}$. The transition matrix of
 the reversed chain remains the same, hence $\sigma_i'=\sigma_i$,
 $i= 1,2$. Thus, the obtained bounds are
\begin{align*}
 \P(Y_1=v_1(x^n)|X^n=x^n)&\geq {\e_1^2\over \e_1^2
 +(1-\e_1)^2},\quad \P(Y_n=v_n(x^n)|X^n=x^n)\geq {\e_1^2\over \e_1^2
 +(1-\e_2)^2},\\
\P(Y_t=v_t(x^n)|X^n=x^n)&\geq {\e_1^4\over \e_1^4
 +(1-\e_1)^2(1-\e_2)^2},\quad t=2,\ldots,n-1.
 \end{align*}
Note that when $\e_1=\e_2=0.5$, then the underlying Markov chain consists of iid Bernoulli random variables with parameter
$0.5$. In this case the Viterbi and the PMAP-alignment are the same.
Given that the ties are broken in favor of 1, $v_t(x^n)=1$ if and only if $f_1(x_t)\geq f_2(x_t)$ .
All the bounds above equal ${1\over 2}$, which is clearly a tight bound.
Without loss of generality, let $v_t(x^n)=1$.
The classification probability in this trivial case can be calculated as
$$\P(Y_t=v_t(x^n)|X^n=x^n)=\P(Y_t=1|X_t=x_t)={f_{1}(x_t)\over
f_{1}(x_t)+f_{2}(x_t)}\geq {1\over 2}.$$
%
%-----------------------------
\subsection{General case}
%--------------------------------
The proof of Proposition \ref{pos} holds only in the case of transition matrices
with non-zero entries. The following counterexample shows that
if the transition matrix contains zeros, a data-independent lower bound 
to the classification probabilities does not exist.
%-----------------
\subsubsection{Counterexample}\label{counterex} Consider a 4-state model with the transition matrix
and initial distribution given by
$$\mathbb P=\left(
    \begin{array}{cccc}
      {1\over 2} & {1\over 2} & 0 & 0 \\
      {1\over 4} & {1\over 4} & {1\over 4} & {1\over 4} \\
      0 & {1\over 3} & {1\over 3} & {1\over 3}  \\
      0 & {1\over 3} & {1\over 3} & {1\over 3} \\
    \end{array}
  \right), \quad \pi=(1/4,1/4,1/4,1/4)'.
$$
Suppose the emission distributions are all discrete, hence $\mu$ is
counting measure and $f_i(x)$, $i=1,\ldots,4$, are all
probabilities. Suppose there exist atoms $x$ and $y$ so that
emission probabilities satisfy the following conditions:
\begin{itemize}
  \item[1)] $f_2(x)=0$, $f_1(x)=f_3(x)=f_4(x)=A>0$,
  \item[2)] $f_1(y)=f_3(y)=f_4(y)=0$, $f_2(y)=D>0$.
\end{itemize}
Let $\e>0$ be arbitrary. We shall show that for $n$ big enough,
there exists a sequence of observations $x^n$ with $p(x^n)>0$,
such that for some time point $t$,
%---------------------
\begin{equation}\label{contra}
\P(Y_t=v_t(x^n)|X^n=x^n)<\e.
\end{equation}
Let $m\in \mathbb{N}$ be so big that
%\begin{equation}\label{defm}
\[{1\over 1+\big({4\over 3}\big)^m}<\epsilon.\]
%\end{equation}
Consider a sequence of observations $x_1,\ldots,x_n$, $n>m$, such
that $x_1=x_2=\ldots=x_m=x$ and $x_{m+1}=y$, where $x$ and $y$ are
the defined atoms. By assumptions, the probability of having
such observations is strictly positive. Let the rest of the
observations, that is $x_{m+2},\ldots,x_n$, be arbitrary with the only
requirement that the probability of emitting $x^n$ is positive, i.e.
$p(x^n)>0$.
Note that since all the paths with positive posterior probability,
including the Viterbi path, pass state 2 at time $m+1$, then
$$\P(Y_t=s|X^n=x^n)=\P(Y_t=s|X^{m+1}=x^{m+1}),\quad t=1,\ldots,m+1.$$
Observe also that any path passing state 2 
before time point $m+1$ will have zero posterior probability. Hence,
the only path passing state 1 at any $t\leq m$ is the path that is
constantly in state 1 up to time $m$. Therefore, for any
$t=1,\ldots, m$,
$$\P(Y_t=1,X^{m+1}=x^{m+1})=\Big({1\over 4}\Big)A^m\Big({1\over 2}\Big)^{m}D=A^m\Big({1\over 2}\Big)^{m+2}D.$$
Note also that there is no path with transition from state 3 or 4
into state 1 that would have positive posterior probability. Hence,
for $s=3,4$ and for any $t\le m$,
 $$\alpha(x^t,s)=2^{t-1}\Big({1\over 4}\Big)A^t\Big({1\over 3}\Big)^{t-1},$$
implying that
$$ \P(Y_m=s,X^{m+1}=x^{m+1})=2^{m-1}\Big({1\over 4}\Big)A^m\Big({1\over 3}\Big)^{m}D.$$
It follows that
$$\P(Y_{m}=1|X^n=x^n)={A^m\big({1\over 2}\big)^{m+2}\over A^m\big({1\over 2}\big)^{m+2}+2^m\big({1\over 4}\big) A^m\big({1\over 3}\big)^m}
={1\over 1+\big({4\over 3}\big)^m}<\epsilon.$$
%----------------------
Thus, if $v_m(x^n)=1$, then (\ref{contra}) holds. Let us show that
up to time point $m$, the Viterbi alignment is given by $v_1=v_2=\cdots =
v_m=1$. Since the Viterbi path passes state 2 at $m+1$, by
optimality principle the observations $x_{m+2},\ldots,x_n$ do not
affect the alignment up to $m+1$. Therefore, it is sufficient to
consider the joint likelihood up to $m+1$. For $u_1=\cdots = u_m=1$,
$u_{m+1}=2$,
$$p(x^{m+1},u^{m+1})=\Big({1\over 4}\Big)\Big({1\over 2}\Big)^{m}A^mD.$$
All other paths  with positive posterior probability can up to
time $m$ pass states 3 and 4 only. For any such path $s^{m+1}$,
$$p(x^{m+1},s^{m+1})=\Big({1\over 4}\Big)\Big({1\over 3}\Big)^{m}A^mD,\quad s_t\in \{3,4\},\quad t=1,\ldots, m,\quad s_{m+1}=2.$$
Since ${1\over 2}>{1\over 3}$, we have
$p(x^{m+1},u^{m+1})>p(x^{m+1},s^{m+1})$, and therefore
$v(x^{m+1})=u^{m+1}$. Thus, (\ref{contra}) holds for $t=m$.
%
%
%-------------
\subsubsection{Data-dependent lower bound}
%-----------
\paragraph{Cluster assumption.}
We shall relax the assumption of positive transitions by the
following much weaker assumption. Let $G_j$ denote the support of the emission distribution $P_j$.
%------------------
We call a subset $C\subset S$  {\it a cluster} if the following
conditions are satisfied:
%\begin{align*}
$$\min_{j\in C}P_j(\cap _{s\in C}G_s)>0 \quad{\rm and}\quad\max_{j\not\in C}P_j(\cap _{s\in C}G_s)=0.$$
%\end{align*}
%---------
Hence, a cluster is a maximal subset of states such that $G_C:=\cap
_{s\in C}G_s$, the intersection of the supports of the corresponding
emission distributions, is  `detectable'. Distinct clusters need not
be disjoint and  a cluster can consist of a single state. In this
latter case such a state is not hidden, since it is exposed by any
observation it emits. If $K=2$, then $S$ is the only  cluster
possible, because otherwise the underlying Markov chain would cease
to be hidden.
%--------------------
The existence of $C$ implies the existence of a set $\X \subset
\cap_{s\in C }G_s$ and $\e>0$, $M<\infty$,
 such that $\mu (\X)>0$, and $\forall x\in
{\X}$ the following statements hold: (i) $\e<\min_{s\in C} f_s(x)$;
(ii) $\max_{s\in C} f_s(x)<M$; (iii) $\max_{s\not \in C} f_s(x)=0$.
For proof, see \cite{IEEE}.\\\\
%---------------
{\bf A1 (cluster-assumption):} There exists a cluster $C\subset S$
such that the sub-stochastic matrix $R=(p_{ij})_{i,j\in C}$ is
primitive, i.e.~there is a positive integer $r$ such that the $r$-th
power of $R$ is strictly positive.\\\\
The cluster assumption {\bf A1} is often met in practice. It is
clearly satisfied if all the elements of $\mathbb P$ are
positive. Since any irreducible aperiodic matrix is primitive,
assumption {\bf A1} is also satisfied if the densities $f_s$ satisfy
the following condition: for every $x\in {\cal X}$, $\min_{s\in
S}f_s(x)>0$, i.e.~for all $s\in S$, $G_s={\cal X}$. Thus, {\bf A1}
is more general than the {\it strong mixing condition} (Assumption
4.2.21 in \cite{HMMraamat}) and also weaker than Assumption 4.3.29
in \cite{HMMraamat}. Note that {\bf A1} implies the aperiodicity of
$Y$, but not vice versa.
\\\\
%-------------
{\bf Example.} Let us reconsider the counterexample in Subsection
\ref{counterex}. The example is very easy to modify so that {\bf A1}
holds. It suffices to have one atom, say $z$, so that $f_j(z)>0$ for
every $j=1,2,3,4$. Then  $z\in \cap_{j\in C} G_j$ so that the
cluster consists of all states, i.e. $C=\{1,2,3,4\}$. Note
that $\mathbb{P}^2$ is primitive, so that $r=2$. The set $\X$ can be
taken as  $\{z\}$.
\\\\
%---------------
Let $x^n$ be fixed and $\X$ and $r$ be as in {\bf A1}. Define for
any $t\in \{1,\ldots,n\}$,
\begin{align*}
w_t(x^n)&:=\min\{t+r<w\leq n: x^{w}_{w-r}\in \X^{r+1}\}\wedge n,\\
u_t(x^n)&:=\max\{ 1\leq u<t-r: x_{u}^{u+r}\in \X^{r+1}\}\vee 1,
\end{align*}
where minimum over the empty set is set to $\infty$ and maximum over the empty set
is set to $-\infty$. Thus, $w_t$ is the first time after $t$ when a word
from $\X^{r+1}$ is fully observed, and $w_t=n$ if there is no such
word up to time $n$. Similarly, $u_t$ is the last time before $t$ when
a word from $\X^{r+1}$ is fully observed, and $u_t=1$ if there is no such
word up to time $n$. The following lemma follows from Proposition
4.1 and Corollary 4.1 in \cite{kuljus}.
%--------------
\begin{lemma}\label{lemmauw} There exist constants $c>0$ and $0<A<\infty$ such that
for every $t=1,\ldots,n$,
\begin{equation}\label{uw}
\P(Y_t=v_t(x^n)|X^n=x^n)\geq c\exp[-A(w_t-u_t)].
\end{equation}
\end{lemma}
The bound in (\ref{uw}) depends on $x^n$, because $w_t$ and $u_t$
depend on $x^n$. If there is no word from  $\X^{r+1}$ in the
observation sequence $x^n$, then $w_t=n$ and $u_t=1$, so that the bound is
$c\exp[-A(n-1)]$, and as we already know such a bound trivially
holds. Hence, (\ref{uw}) clearly improves the trivial bound given in 
(\ref{classificationV}).
%
%
%--------------
\paragraph{Stochastic bounds that are independent of $n$.}
Letting now the data $X^n$ be random, we get that $W_t$ and $U_t$
are random stopping times, and the bound in (\ref{uw}) can be
written as
\begin{equation}\label{UW}
\P(Y_t=v_t(X^n)|X^n)\geq c\exp[-A(W_t-U_t)].
\end{equation}
Let us study the distribution of the random variables $W_t-U_t$.
Obviously, $W_t-U_t \leq n-1$, and the distribution of $W_t-U_t$
depends on both $t$ and $n$. We would, however, like to have an
upper bound on $W_t-U_t$ that is independent of $n$ and, if possible, also
independent of $t$. Consider the observation process $X_1,X_2,\ldots$,
and let
$$W^*_t:=\min \{w>t+r: X^{w}_{w-r}\in \X^{r+1}\}.$$
Thus, $W_t=W_t^*\wedge n$, so that $W_t\leq W^*_t$. The random
variable $W^*_t$ is independent of $n$, and as the following
proposition shows, $W^*_t-t$ has exponential
tail that can be chosen independently of $t$.
%------------------
\begin{proposition}\label{geom}
Assume {\bf A1}. There exist constants $a>0$ and $b>0$ such that for
any initial distribution $\pi$ and for any $t$,
%\begin{equation}\label{expo}
\[\P(W^*_t-t>k)\leq a\exp[-bk],\quad k=0,1,2,\ldots.\]
%\end{equation}
\end{proposition}
%----------
%---------------------------------------
%----------------
The proof is given in Appendix. Because of the proposition,
the following corollary holds.
%-------------
\begin{corollary}\label{cor1}
Assume {\bf A1.} Then for any initial distribution, the following lower
bound holds:
%\begin{equation}\label{w}
\[\P(Y_t=v_t(X^n)|X^n)\geq c\exp[-At]Z_t. \]
%\end{equation}
Here $Z_t$ is a $\sigma(X_1,X_2,\ldots)$-measurable random variable
such that $-\ln Z_t$ has exponential tail independent of $t$, that is for some
positive constants $r$ and $d$ and for every $u>0$, $\P(-\ln Z_t>u)\leq r\exp[-du]$.
\end{corollary}
%-------------------
\begin{proof} From (\ref{UW}) it follows that
\begin{equation*}\label{UW2}
\P(Y_t=v_t(X^n)|X^n)\geq
c\exp[-A W^*_t]=c\exp[-At]\exp[-A(W^*_t-t)]=c\exp[-At]Z_t,
\end{equation*}
where $Z_t=\exp[-A(W^*_t-t)]$. Thus, $-\ln Z_t=A(W^*_t-t)$, and for
any $u>0$,
\begin{align*}
\P(-\ln Z_t>u)&=\P(W^*_t-t>A^{-1}u)=\P(W^*_t-t> \lfloor A^{-1}u
\rfloor)\leq a\exp[-b \lfloor A^{-1}u \rfloor]\\
&\leq  a\exp[-b( A^{-1}u-1)]=r\exp[-du],
\end{align*}
where $r:=a e^b$ and $d:=bA^{-1}$. \end{proof}
%--------------
\paragraph{Stationary case.} Let now the initial distribution be
stationary. Then it is convenient to embed $X$ into a two-sided
stationary hidden Markov process $\{X_t\}_{t=-\infty}^{\infty}$.
Now, besides the stopping time $W_t^*$, we can also define the time  $U_t^*$ as follows:
$$U^*_t:=\max \{u<t-r:
X_{u}^{u+r}\in \X^{r+1}\}.$$ Thus, $U_t=U^*_t \vee 1$, so that
$U_t\geq U^*_t$. Proposition \ref{geom}, possibly with some other
constants, holds also for $t-U^*_t$.  Therefore, for any $t$,
the random variable $W^*_t-U_t^*$ has exponentially decreasing
tail:
\begin{equation*}
\P(W^*_t-U^*_t>k)=\P\big((W^*_t-t)+(t-U^*_t)>k\big)\leq
\P(W^*_t-t>{k\over 2})+\P(t-U^*_t>{k\over 2})\leq a_o e^{-b_o k},
\end{equation*}
where $a_o$ and $b_o$ are some positive constants. Thus, we have the
following lower bound.
\begin{corollary}\label{cor2}
Assume {\bf A1} and let the initial distribution $\pi$ be
stationary. Then
\begin{equation}\label{w}
\P(Y_t=v_t(X^n)|X^n)\geq Z_t,
\end{equation}
where $Z_t$, $t=1,\ldots,n$, are $\sigma(\{X_t\}_{t=-\infty}^{\infty})$-measurable identically
distributed random variables such that $-\ln Z_t$ has exponential tail, that is for some positive constants $r$ and $d$
and for every $u>0$, $\P(-\ln Z_t>u)\leq r\exp[-du]$. Hence, $E[-\ln Z_t]<\infty$.
\end{corollary}
%--------
\begin{proof} From (\ref{UW}) it follows that
$$\P(Y_t=v_t(X^n)|X^n)\geq c\exp[-A(W^*_t-U^*_t)]=:Z_t.$$
By stationarity, the random variables $Z_t$ are identically
distributed.  The rest of the proof is the same as the one of
Corollary \ref{cor1}.\end{proof}\\\\
%---------------------
%
%
The accuracy of the Viterbi alignment, that is the expected number of correctly
classified states given $X^n$, is for a stationary chain according to (\ref{w}) bounded below by
$\sum_{t=1}^n Z_t$. Therefore, for a stationary chain, we can with help of Corollary \ref{cor2}
find an upper bound for the probability that the accuracy is
less than $a_n$, where $a_n<n$. Let $M=E[-\ln Z_t]<\infty$. Then
%-------------
\[ \P\Big(\sum_{t=1}^n Z_t \leq a_n\Big)= \P\Big(-\ln \big({1\over n}\sum_{t=1}^n Z_t \big) \geq -\ln {a_n\over n}\Big) \]
\[ \stackrel{\mbox{\textit{(Jensen)}}}{\leq} \P\Big({1\over n} \sum_{t=1}^n \big(-\ln Z_t \big) \geq \ln {n\over a_n}\Big)
 \stackrel{\mbox{\textit{(Markov)}}}{\leq} { E\big[-\ln Z_t \big] \over \ln{({n\over a_n})}}={M\over \ln{({n\over a_n})} }.\]
%
%------------------------
\section{Iterative algorithm}\label{sec:algorithm}
%----------
Recall that we aim to improve the accuracy of the Viterbi alignment.
Since the accuracy is just the sum of classification probabilities,
the straightforward idea for doing this is to find the time points with lowest
classification probabilities, replace them by the PMAP-states (or by the
true states when peeping is possible), and replace the original Viterbi
alignment by the restricted Viterbi alignment. As explained in the
introduction, such a bunch approach has a big drawback, since
typically the time points with low classification probabilities are situated 
next to each other. Therefore, substituting a number of consecutive states with the corresponding PMAP-states can make the
adjusted path inadmissible. The following iterative algorithm
ensures that the adjusted alignment remains admissible. 
%
%-----------------------
\subsection{Description of the iterative algorithm}
%------------------
\begin{description}
  \item[Input:] observations $x^n$, a threshold parameter $\delta$, $0<\delta<{1\over K}$, and the maximum number of iterations $M$.
  \item[Initialization:] find the Viterbi alignment $v(x^n)$ and calculate the  classification probabilities
  $$\rho_t^{(0)}:=\P(Y_t= v_t|X^n=x^n),\quad t=1,\ldots, n.$$ Define $v^*:=v$.
%--------------------

%------------------
\item[For $m=1,\ldots,M$ do:]
 if $\min_t \rho_t^{(m-1)}\geq \delta$, then quit, else 
\begin{itemize}
  \item[1)] find the time point $t_m$ with lowest conditional classification
  probability and the state $w_m$ that maximizes the corresponding conditional classification probability:
  $$t_m:=\arg\min\{\rho_t^{(m-1)}:t=1,\ldots,n\},$$
  $$w_{m}:=\arg\max_{s\in S}\P(Y_{t_m}=s|X^n=x^n;
  Y_{t_i}=w_{i},i=1,\ldots,m-1);$$
  \item[2)] let $S^n(m):=\{s^n\in S^n: s_{t_1}=w_1,\ldots,s_{t_m}=w_{m}\}$, find the new restricted Viterbi path $v^{(m)}$,
   \[v^{(m)}:=\arg\max_{s^n\in S^n(m)}\P(Y^n=s^n|X^n=x^n) \] 
   \[=\arg\max_{s^n}\P(Y^n=s^n|X^n=x^n;Y_{t_i}=w_i,i=1,\ldots,m),\]
   define $v^*:=v^{(m)}$;
   \item[3)] calculate the new conditional classification probabilities $\rho_t^{(m)}$,
   \begin{equation} \label{restr_classprob}
   \rho_t^{(m)}:=\P(Y_t=  v^{(m)}_t|X^n=x^n;Y_{t_i}=w_i,i=1,\ldots,m),\quad t=1,\ldots, n. \end{equation}
\end{itemize}
\item[Output:] the alignment $v^*(x^n)$.
%-----------------
\end{description}
In the algorithm described above, thus, at first the time $t_1$ 
with the lowest classification probability is found. Then, at
this point, the state $w_1$ with maximum posterior probability -- the
PMAP state -- is found. The state $w_1$ at time point $t_1$ is taken
as it were the true state, and in all what follows, only the paths
passing $w_1$ at $t_1$ are considered. The conditional
classification probabilities in the next step are computed given
the event $\{Y_{t_1}=w_1\}$. The time $t_2$ has the smallest
conditional classification probability and the state $w_2$ is the
state that at $t_2$ has the maximum posterior probability given
$\{Y_{t_1}=w_1\}$. This means that the probability
$\P(Y_{t_1}=w_1,Y_{t_2}=w_2|X^n=x^n)$ is strictly positive, thus 
the algorithm guarantees that the alignment remains admissible, 
i.e.~it has positive posterior probability. In what follows, the states
$w_1$ and $w_2$ at time points $t_1$ and $t_2$ are taken as they were true
states, and  all probabilities are calculated conditional on 
$\{Y_{t_1}=w_1,Y_{t_2}=w_2\}$. The output $v^*$
has then always positive posterior probability that decreases as $m$
increases, because at every step of
iterations, an additional constraint is imposed.\\\\
%-------------
As explained in the introduction, another problem with the 
bunch approach is that replacing the states with low classification
probability by the PMAP-states can change the alignment, so that the
classification probabilities of the restricted Viterbi alignment can drop
below the threshold somewhere else. As the example in the next
subsection shows, this can indeed be the case. The iterative
algorithm does not necessarily exclude such possibility, but we
have a reason to believe that such a phenomenon is less likely to
happen. The reasoning is as follows. As is shown in
\cite{IEEE,AVT4,K2}, (under some conditions) the influence of
changing the Viterbi alignment is local. This means that (with high
probability) there exist time points $1=u_0<u_1<u_2<\cdots<u_k=n$, so that
if $t\in (u_{j-1},u_j)$, then forcing the alignment to pass a
prescribed state at time $t$ changes the Viterbi alignment in the
range $(u_{j-1},u_j)$ only (see also \cite{intech}). Thus, the
influence of adjusting the alignment at $t$ is local. Suppose now
that at some $t\in (u_{j-1},u_j)$, the classification probability
$\P(Y_t= v_t|X^n=x^n)$ is very low. Then as explained before, 
the classification probability is most likely low also for the
neighbors, meaning that the behaviour of the Viterbi alignment in 
$(u_{j-1},u_j)$ is atypical, so the piece $(u_{j-1},u_j)$ is somehow abnormal.
Changing the alignment at $t$ changes it also in the
neighborhood of $t$, but not outside of the piece $(u_{j-1},u_j)$. It is
meaningful to believe that the abnormal piece is now adjusted, so
that the classification probabilities of the adjusted alignment
$v^{(1)}$ are bigger not only at $t$ but also in the neighborhood.
%In other words, there is a reason to believe that with one change
%only, the whole abnormality of $(u_{j-1},u_j)$ can be taken care of.
This is the reason why the iterative algorithm achieves the same
effect as the bunch approach with a considerably smaller number of replacements.
If there is now another abnormal piece $(u_{l-1},u_l)$ ($l\ne j$),
then the previous changes do not influence the Viterbi alignment in
that piece, so that at some $t_2\in (u_{l-1},u_l)$, the
(unconditional) classification probability of $v^{(1)}$ is still
atypically low. The question is whether the algorithm still finds $t_2$, since it uses the
conditional (given $\{Y_{t_1}=w_1\}$) smoothing probabilities.
However, for many models the smoothing probabilities $\P(Y_t=s|X^n)$
have the so-called exponential forgetting probabilities
\cite{PMAP,PMAP2,ungarlased}, so that for some constant $0<\rho_o<1$, for a 
non-negative finite random variable $C$ and for any state $s$, 
$$\Big|\P(Y_{t_2} = s|X^n=x^n)-\P(Y_{t_2} = s|X^n=x^n,Y_{t_1}=w_1)\big|\leq C\rho_o^{|t_1-t_2|}.$$ 
This inequality implies that when $t_1$ and $t_2$ are sufficiently far
from each other, then the conditioning on $\{Y_{t_1}=w_1\}$ does not influence much the
 classification probability at $t_2$, and the algorithm finds the
 next abnormal piece. For a similar result, see Corollary 2.1 in \cite{PMAP}.
 %---------------------
 \\\\
 If peeping is possible, then instead of revealing a bunch of true states at once, one can also perform peeping iteratively. 
 Although (computationally) more costly, the iterative way of
 adjusting the Viterbi alignment has several advantages over the bunch approach. The iterative algorithm tends to adjust the Viterbi
 alignment piecewise. Since the number of abnormal pieces is usually smaller than the number of time points with low
 classification probability, the number of replacements (iterations) needed to reach a certain effect is considerably smaller for the iterative approach 
 compared to the bunch approach.
 %-----------------------------------
\subsection{Comparison of the bunch and iterative approach}
%-----------
\subsubsection{A case study}\label{subsec:case}
In this example, we consider a model that is used in \cite{seg} for illustrating the task of predicting protein secondary structure in single amino-acid sequences. 
The underlying Markov chain has six possible states. The transition matrix and initial distribution are 
as follows: 
\[\mathbb P = \left( \begin{array}{cccccc}
      0.8360 & 0.0034 & 0.1606 & 0 & 0 & 0 \\
      0.0022 & 0.8282 & 0.1668 & 0.0028 & 0 & 0 \\
      0.0175 & 0.0763 & 0.8607 & 0.0455 & 0 & 0 \\
      0 & 0 & 0 & 0.7500 & 0.2271 & 0.0229 \\
      0 & 0 & 0 & 0 & 0.8450 & 0.1550 \\
      0 & 0.0018 & 0.2481 & 0 & 0 & 0.7501 \\
      \end{array} \right) , \]
\[\quad \pi= (0.0016,0.0041,0.9929,0.0014,0,0)' \, . \]
Many transitions are impossible and this can make a PMAP-sequence 
inadmissible. The observations come from a 20-symbol
emission alphabet of amino-acids, the emission matrix is given in
Appendix. In order to compare the bunch approach and the iterative approach, we have
generated an observation sequence (together with the underlying Markov chain) of length $n=1000$ from this model. 
We shall compare the two approaches for both PMAP-replacements and peeping.
\\\\
%------------------------
To compare the behaviour of the bunch and iterative algorithm, we provide for both
algorithms a table with some summary characteristics that have been
calculated for different number of replacements or iterations $m$,
respectively. The simulation results are given in Tables 1 -- 4. In these tables, 
\textit{Errors} denotes the real number of classification errors and \textit{E(Errors)} the expected number of classification errors, 
$\rho_{min}^{uncond}:=\min_t \P(Y_t=v_t^{(m)}|X^n=x^n)$ and $\rho_{min}^{cond}:=\min_t \rho_t^{(m)}$ (see (\ref{restr_classprob}))
give respectively the minimum unconditional and conditional classification probability for the restricted alignment
after $m$ replacements/iterations, and \textit{Log-likelihood} gives the logarithm of the posterior probability of the restricted alignment.   
Observe that \textit{Errors} depends on the realization of the underlying
hidden Markov chain. The PMAP-alignment of the generated sequence has 467
classification errors  and it is inadmissible, i.e.~its posterior
probability is zero. The Viterbi alignment has 481 classification
errors.
\\\\
Suppose that the threshold parameter $\delta$ is set to $0.1$. There are $140$ classification probabilities smaller than $0.1$ for the Viterbi alignment of this sequence.
Using the bunch algorithm would mean that we substitute the states corresponding to these 140 low probabilities with the respective
PMAP-states, and find then the restricted Viterbi path. From Table 1 we can see
that the likelihood of the restricted path is zero. The likelihood of the restricted Viterbi will be zero after 78 replacements.
This depends on replacement of many consecutive states: all the states from time point 712 to 754, except at 728, are substituted,
whereas from 753 to 754 we obtain an inadmissible transition $3\rightarrow 5$. If we would use the iterative algorithm
with the same threshold instead, we would stop after 18 iterations because $\min_t \rho_t^{(18)}=0.1094$. The number of classification errors
for the restricted alignments obtained with the bunch algorithm (140 replacements) and iterative algorithm (18 iterations) are
486 and 485, respectively. The 11 lowest unconditional classification probabilities for the restricted alignments are:
\begin{itemize}
\item[1)] Bunch 0.0448, 0.0449, 0.0474, 0.0506, 0.0558, 0.0655, 0.0671, 0.0771, 0.0880, 0.0944, 0.1018;
\item[2)] Iterative 0.1094, 0.1149, 0.1184, 0.1227, 0.1247, 0.1276, 0.1305, 0.1383, 0.1426, 0.1428, 0.1460.
\end{itemize}
We can see that in the case of bunch algorithm, after fixing the
preliminary set of 140 time points, the classification probability
has dropped below $\delta$ for ten time points. For the iterative
algorithm, all the probabilities are above the threshold. 
\\\\
Recall that the unrestricted Viterbi alignment has 481 classification errors.
In Tables 1--2 we can see how the number of
classification errors decreases at first with increasing number of replacements/iterations, but then it starts
to increase again.  The minimum number of classification errors for the restricted alignments is 428.
The iterative algorithm reaches this number after four iterations. To obtain the same error rate with the bunch
algorithm, we need to make 37 replacements. The likelihood of the restricted Viterbi alignment after four iterations
is higher compared to the likelihood of the restricted sequence obtained after 37 substitutions with the bunch algorithm
(log-likelihoods are $-171.28$ and $-172.55$, respectively). This shows that the iterative algorithm is more effective since it works piecewise.
If we would use the bunch algorithm with four replacements,
the replacements would occur at time points 723, 724, 725 and 733, which give the four lowest classification probabilities.
This means that we would make adjustments at three consecutive time points. With the iterative algorithm,
the substitutions would be made at 723, 752, 582 and 557, i.e.~the problematic pieces are fixed in turn. With the iterative algorithm, the available
information for making adjustments is used more efficiently.
%------------
\\\\
Observe that  E(Errors) is just $n$ minus the accuracy. For the 
Viterbi alignment this number is 544. The best possible expected
number of errors, which corresponds to the PMAP-alignment, is 459. Again,
to reach a certain decrease in the expected number of errors, 
the iterative algorithm needs a smaller number of replacements than the bunch algorithm. 
After ten replacements/iterations for example, E(Errors) is 522 (bunch) and 501 (iterative). 
To achieve ${\rm{E(Errors)}}=497$, 15 iterations are needed, whereas the bunch algorithm requires about 70 replacements. 
The decrease from 544 to 497 might not seem that big, but one should take into consideration that the maximum 
possible improvement is $544-459=85$. Hence, the improvement $544-497=47$ that the 
iterative algorithm achieves with 15 replacements, is more than half of the possible improvement. 
\vskip 1\baselineskip\noindent
%--------------------------------------
{{\bf{Table 1}}. PMAP-replacements with the bunch algorithm.}
\begin{center} 
\begin{tabular}{|r|cccc|}
  \hline
  $m$ & Errors & E(Errors) &  $\rho_{min}^{uncond}$ & Log-likelihood \\
  \hline
   1  &  452  & 528 &     0.0279   &     -168.58 \\
   2  &  452  & 528 &     0.0279   &     -168.58 \\
   3  &  452  & 528 &     0.0279   &     -168.58 \\
   4  &  452  & 528 &     0.0279   &     -168.58 \\
   5  &  452  & 528 &     0.0279   &     -168.58 \\
  10  &  449  & 522 &     0.0437   &     -169.44 \\
  15  &  445  & 522 &     0.0437   &     -172.13 \\
  20  &  445  & 522 &     0.0437   &     -172.13 \\
  25  &  445  & 522 &     0.0437   &     -172.18 \\
  30  &  445  & 522 &     0.0437   &     -172.18 \\
  35  &  433  & 519 &     0.0448   &     -172.50 \\
  37  &  428  & 517 &     0.0448   &     -172.55 \\
  40  &  429  & 516 &     0.0448   &     -172.80 \\
  50  &  455  & 508 &     0.0448   &     -175.39 \\
  60  &  461  & 505 &     0.0448   &     -177.65 \\
  70  &  487  & 496 &     0.0448   &     -177.89 \\
  77  &  483  & 494 &     0.0448   &     -178.90 \\
  78  &  483  & 494 &     0.0448   &      $ -\infty$ \\
  140 &  486  & 488 &     0.0448   &      $-\infty$ \\
  \hline
\end{tabular}
\end{center}
{{\bf{Table 2}}. PMAP-replacements with the iterative algorithm.}
\begin{center}
\begin{tabular}{|r|ccccc|}
  \hline
  $m$ & Errors & E(Errors) & $\rho_{min}^{cond}$  & $\rho_{min}^{uncond}$ & Log-likelihood \\
  \hline
    1  &     452  & 528 &  0.0279   &   0.0279   &     -168.58 \\
    2  &     451  & 523 &  0.0437   &   0.0437   &     -169.37 \\
    3  &     439  & 520 &  0.0439   &   0.0448   &     -169.69 \\
    4  &     428  & 515 &  0.0103   &   0.0458   &     -171.28 \\
    5  &     433  & 512 &  0.0453   &   0.0458   &     -172.64 \\
   10  &     452  & 501 &  0.0451   &   0.0576   &     -176.19 \\
   15  &     458  & 497 &  0.0459   &   0.0608   &     -179.16 \\
   18  &     485  & 487 &  0.1094   &   0.1094   &     -181.85 \\
   77  &     498  & 481 &  0.2779   &   0.0947   &     -215.19 \\
   78  &     502  & 481 &  0.3105   &   0.0947   &     -215.38 \\
  \hline
\end{tabular}
\end{center}
%------------
Tables 3 and 4 compare the bunch and the iterative approach in the case of 
peeping. In this case, we take into
account the additional information obtained when revealing states. Thus, E(Errors)
is calculated with help of conditional classification probabilities:
\begin{equation} \label{condE}
{\rm {E(Errors)}}=n-\sum_{t=1}^n \P(Y_t= v^{(m)}_t |X^n=x^n,Y_{t_1}=y_{t_1},\ldots,Y_{t_m}=y_{t_m}).
\end{equation}
Again, the iterative algorithm is more efficient than the 
bunch algorithm. After 78 replacements with the bunch approach, the minimum (conditional) classification probability 
for the restricted sequence is still 0.0452. For iterative peeping, this probability is 0.1256 after 
10 iterations. The first replacement has a big positive effect: the
number of errors decreases from 481 to 452 (apparently a whole piece
is corrected). But the subsequent replacements with the bunch method have either a negative effect (causing thus additional
errors) or give an additional decrease in the number of errors that is generally smaller than the
number of replacements $m$. As Table 4 shows, adjusting the
alignment iteratively is much more efficient in this sense, since $m$ additional
replacements after the first one decrease the number of errors by more than $m$. 
The number of errors for $m=3$ and $m=4$ in Table 4 shows that iterative peeping can also 
have a negative effect. We can also study the effect of the iterative approach when states are substituted 
with the corresponding PMAP-states or true states. Table 2 and Table 4 show that after 15 iterations for example, 
the restricted sequence has 458 and 395 errors when replacements are done with the PMAP-states or true states, respectively.  
%
%----------------
Note that E(Errors) might increase with $m$ (see Table 3). We shall address this 
issue more closely in Section \ref{sec:peep}.
\newline \newline
{{\bf{Table 3}}. Peeping with the bunch algorithm.}
\begin{center}
\begin{tabular}{|r|cccc|}
  \hline
  $m$ & Errors & E(Errors) & $\rho_{min}^{cond}$  & Log-likelihood \\
  \hline
    1 & 452 & 527 & 0.0279 &  -168.58 \\
    2 & 485 & 516 & 0.0319 &  -170.52 \\
    3 & 485 & 515 & 0.0319 &  -170.52 \\
    4 & 451 & 527 & 0.0238 &  -175.24 \\
    5 & 450 & 524 & 0.0238 &  -175.29 \\
   10 & 442 & 512 & 0.0437 &  -181.32 \\
   15 & 436 & 506 & 0.0437 &  -187.53 \\
   20 & 435 & 505 & 0.0437 &  -187.92 \\
   25 & 429 & 503 & 0.0437 &  -189.43 \\
   30 & 429 & 501 & 0.0437 &  -189.43 \\
   35 & 416 & 491 & 0.0439 &  -190.28 \\
   37 & 423 & 484 & 0.0439 &  -191.26 \\
   40 & 423 & 483 & 0.0394 &  -191.26 \\
   50 & 415 & 457 & 0.0447 &  -192.12 \\
   60 & 408 & 450 & 0.0452 &  -193.28 \\
   70 & 406 & 435 & 0.0452 &  -193.50 \\
   77 & 404 & 429 & 0.0452 &  -194.61 \\
   78 & 404 & 429 & 0.0452 &  -194.61 \\
  140 & 369 & 383 & 0.1094 &  -215.54 \\
\hline
\end{tabular}
\end{center}
%
%\vskip 1\baselineskip \noindent
%
\newpage \noindent
{{\bf{Table 4}}. Peeping with the iterative algorithm.}
\begin{center}
\begin{tabular}{|r|cccc|}
  \hline
  $m$ & Errors & E(Errors) & $\rho_{min}^{cond}$  &  Log-likelihood \\
  \hline
    1 &  452 & 527 & 0.0279 &  -168.58 \\
    2 &  448 & 514 & 0.0437 &  -170.36 \\
    3 &  436 & 506 & 0.0439 &  -170.69 \\
    4 &  423 & 495 & 0.0458 &  -173.18 \\
    5 &  430 & 488 & 0.0484 &  -174.16 \\
   10 &  422 & 445 & 0.1256 &  -179.06 \\
   15 &  395 & 429 & 0.1152 &  -183.16 \\
   18 &  393 & 414 & 0.1435 &  -183.33 \\
   77 &  299 & 299 & 0.3146 &  -228.85 \\
   78 &  298 & 298 & 0.2970 &  -228.97 \\
  \hline
\end{tabular}
\end{center}
%---------------
\subsubsection{Threshold-based adjustments}
In this example, we consider the following two-state hidden Markov
model. The transition matrix and initial probabilities are given by
\[ \mathbb P = \left( \begin{array}{cc}
      0.9 & 0.1 \\
      0.1 & 0.9 \\
      \end{array} \right) ,  \quad \quad  \pi'=(0.5,0.5) ,\]
and the emission distributions are given by $\mathcal N(0,1)$ and
$\mathcal N (0.5,1)$. We have generated 100 observation sequences of
length $n=1000$ from this HMM and studied the mean behavior of the
restricted Viterbi sequences for different threshold parameters $\delta$. We
study  threshold-based adjustments. For the bunch approach this means that for all the time points with
lower classification probability than a given $\delta$, the Viterbi
state is substituted with the corresponding PMAP-state (or in the case of peeping with the true state), and
thereafter restricted segmentation is performed. In the case of
iterative algorithm, replacements are based on conditional
classification probabilities and performed iteratively. For every
restricted alignment, we calculate the real number of classification
errors, the expected number of classification errors, the minimum
conditional and unconditional classification probability, and the
log-likelihood of the restricted Viterbi path. The mean values of
these characteristics over the hundred replicates for the
unrestricted Viterbi are as follows: 350, 354, 0.15 and -105.8. The
average values of the characteristics for the restricted sequences
are given in Tables 5 -- 8. The average number of
substitutions made and its standard deviation can be seen in columns 
\textit{Replacements} and \textit{Iterations} for the bunch and
iterative algorithm, respectively. The average number of PMAP-errors
for the studied sequences is 306. \vskip 1\baselineskip \noindent
%---------------------------------
{{\bf{Table 5}}. PMAP-replacements: mean behavior of the restricted alignments for the
bunch algorithm.}
\begin{center}
\begin{tabular}{|c|ccccc|}
  \hline
  $\delta$ & Replacements & Errors & E(Errors) & $\rho_{min}^{uncond}$ &   Log-likelihood \\
  \hline
    0.20  &   7.50 (5.0)  &  341  &   344  &   0.19 &    -107.4 \\
    0.25  &   19.7 (9.4)  &  338  &   340  &   0.18 &    -109.7 \\
    0.30  &   39.1 (14.6) &  340  &   340  &   0.16 &    -112.6 \\
  \hline
\end{tabular}
\end{center}
\vskip 1\baselineskip \noindent
%---------------------------------
{{\bf{Table 6}}. PMAP-replacements: mean behaviour of the restricted alignments for the
iterative algorithm.}
\begin{center}
\begin{tabular}{|c|cccccc|}
  \hline
  $\delta$ & Iterations & Errors & E(Errors) & $\rho_{min}^{uncond}$ &  $\rho_{min}^{cond}$  &  Log-likelihood \\
  \hline
  0.20   &    3.3 (2.0)  &  336  &  341  &  0.22  &   0.22  & -107.7 \\
  0.25   &    7.4 (3.4)  &  327  &  330  &  0.26  &   0.26  & -110.8 \\
  0.30   &    13.9 (4.9) &  321  &  321  &  0.31  &   0.31  & -115.6 \\
  \hline
\end{tabular}
\end{center}
Compare the bunch and iterative algorithm for $\delta=0.25$, for
example. On average, there are 20 classification probabilities lower
than 0.25. After substituting the states with low classification
probability according to the bunch algorithm, the average minimum
classification probability for the restricted Viterbi alignments is
0.18 and the average number of classification errors is 338. For the
iterative algorithm with the same threshold, we need 7 iterations on
average. The average minimum classification probability for the
restricted alignments is 0.26, which is above the threshold, and the
average number of classification errors is 327. This demonstrates
that the iterative algorithm is more efficient.
\\\\
In the same way, we can compare the threshold-based adjustment
procedure for the bunch and iterative algorithm in the case of
peeping. To take into account the information obtained through
revealing states, we consider conditional probabilities when
calculating the classification probabilities and the expected number
of classification errors for the restricted Viterbi alignments.
\vskip 1\baselineskip \noindent
%-------------------
{{\bf{Table 7}}. Peeping: mean behavior of the restricted alignments
for the bunch algorithm.}
\begin{center}
\begin{tabular}{|c|ccccc|}
  \hline
  $\delta$ & Replacements & Errors & E(Errors) & $\rho_{min}^{cond}$  &  Log-likelihood \\
  \hline
   0.20  &  7.5 (5.0)   &   335  & 339  &  0.19  & -107.8 \\
   0.25  &  19.7 (9.4)  &   324  & 325  &  0.19  & -111.5 \\
   0.30  &  39.1 (14.6) &   307  & 310  &  0.18  & -117.1 \\
  \hline
\end{tabular}
\end{center}
\vskip 1\baselineskip \noindent
%-------------------
{{\bf{Table 8}}. Peeping: mean behavior of the restricted alignments
for the iterative algorithm.}
\begin{center}
\begin{tabular}{|c|ccccc|}
  \hline
  $\delta$ & Iterations & Errors & E(Errors) & $\rho_{min}^{cond}$  &  Log-likelihood \\
  \hline
    0.20  &   3.2 (2.0)  &  333  &  337  &  0.22  &  -107.3 \\
    0.25  &   6.9 (3.2)  &  319  &  322  &  0.26  &  -109.5 \\
    0.30  &   12.2 (4.1) &  304  &  306  &  0.31  &  -112.8 \\
  \hline
\end{tabular}
\end{center}
\vskip 1\baselineskip \noindent
%-------------------
Consider again $\delta=0.25$. When using the bunch algorithm, we
would need to peep at 20 time points on average, whereas with the
iterative algorithm the average number of peepings would be 7. For
the bunch algorithm, the mean minimum classification probability for
the restricted sequences is 0.19, which is below the threshold, and
the average number of errors is 324. The same characteristics in the
case of iterative peeping are 0.26 and 319, respectively.
%-----------------
\section{Unsuccessful peeping}\label{sec:peep}
%-------------
Recall Table 3. With bunch peeping, the number of expected errors E(Errors) for $m=4$
is much bigger than for $m=3$ (527 and 515, respectively). This
means that peeping at four points is much worse than peeping at three
points -- an additional peeping at $t_4$ has a 
negative effect. However, according to (\ref{condE}), E(Errors) when
$m$ hidden states are revealed is conditional on $x^n$ as well as on $y_{t_1},\ldots y_{t_m}$,
implying that the negative effect we see in this example might be due to ``bad'' value of
$Y_{t_4}$ that in our simulations happens to be very untypical. When
taking the expectation over $Y_{t_4}$, the average effect can still be
positive, because the untypical value has very little probability
and for the rest of the values everything is normal. This
speculation arises the following question: is it possible to peep at
some fixed time point, say $t_1$, so that E(Errors) increases also
when averaging over $Y_{t_1}$? Formally, the question is the
following: do there exist an HMM, a sequence of observations $x^n$
having a positive likelihood, and a fixed time point $t_1$ such that the
following inequality holds:
\begin{equation}\label{badpeep}
\sum_{t=1}^n \P(Y_t=v_t(x^n)|X^n=x^n)>\sum_{t=1}^n \P(Y_t=v^{(1)}_t(x^n,Y_{t_1})|X^n=x^n)?
\end{equation}
Here $v^{(1)}$, as previously, stands for the restricted Viterbi
alignment given the value of $Y_{t_1}$. Inequality (\ref{badpeep})
states that the accuracy of the unrestricted Viterbi alignment is
strictly bigger than that of the restricted Viterbi alignment after
peeping $Y_{t_1}$. In what follows,  we present an example showing
that such an {\it unsuccessful peeping} is possible and (\ref{badpeep})
can happen.
%-------------------
%--------------------
\paragraph{The  model and observations.} Consider a 3-state HMM with the transition matrix
 $$\mathbb{P}=\left(
  \begin{array}{ccc}
    {2\over 3}(1-\epsilon) & {2\over 3}\epsilon  & {1\over 3} \\
    {2\over 3}\epsilon  &  {2\over 3}(1-\epsilon) &  {1\over 3} \\
    {1\over 2} & 0 & {1\over 2} \\
  \end{array}
\right),$$ where $0<\epsilon<{1 \over 2}$, implying that ${2\over
3}(1-\epsilon)>{1\over 2}$. Let the initial distribution be
stationary, i.e.
$$\pi_1={3\over 5}\Big({1+2\epsilon\over 1+4\epsilon}\Big),\quad
\pi_2={6\over 5}\Big({\epsilon\over 1+4\epsilon}\Big),\quad \pi_3={2\over 5}.$$
%-------------------
Let $\delta>0$ be so small that
\begin{equation}\label{delta}
(1+\delta) \epsilon<(1-\epsilon)
\end{equation}
%--------------------------------------------
and let  $m\in \mathbb{N}$ be big (will be specified later).
Suppose $x,y,z,a\in {\cal X}$ are such that
\begin{enumerate}
  \item[1)] $f_{1}(x)=1$ and $f_2(x)=f_3(x)=0$;
  \item[2)] $f_2(y)=1+\delta$ and $f_1(y)=f_3(y)=1$;
  \item[3)] $f_3(a)=0$, $f_1(a)=f_2(a)=1$;
  \item[4)] $f_1(z)=f_2(z)=f_3(z)=1$.
\end{enumerate}
Let the observations $x_1,\ldots,x_n$ be as follows: $n=m+2$ and
$$x_1=x,\quad x_2=y,\quad x_3=x_4=\cdots = x_m=z,\quad x_{m+1}=a,\quad x_{m+2}=x.$$
%-------------------
\paragraph{Viterbi alignment.}
By condition 1), all the state paths with positive posterior
probability begin and end in state 1. From (\ref{delta}) it follows
that
$$\Big({2\over 3}(1-\epsilon)\Big)^{m+1}>\Big({2\over
3}(1-\epsilon)\Big)^{m-1}\Big({2\over 3}\epsilon\Big)^{2}(1+\delta),$$
implying that
\begin{align*}
&\P\Big(Y_{1}=\cdots
=Y_{n}=1\Big|X^n=x^n\Big)
>\P\Big(Y_1=1,Y_2=
\cdots=Y_{n-1}=2,Y_n=1|X^n=x^n\Big).\end{align*} From
${2\over 3}(1-\epsilon)>{1\over 2}$ it follows that the posterior
probability to remain in state 1 is bigger than jumping from state 1 to state 3, remaining then there
and jumping thereafter back to state 1. Formally, for any $1\leq k<l<m+1$,
\begin{align*}
&\P(Y_{1}=\cdots =Y_{n}=1\big|X^n=x^n)\\
&>\P\big(Y_{1}=\cdots= Y_k=1,Y_{k+1}=\cdots =Y_l=3,Y_{l+1}=\cdots =Y_{n}=1\big|X^n=x^n\big).\end{align*}
This means that the Viterbi alignment remains in state 1 all the
time.
%-------------------------------
\paragraph{Restricted Viterbi alignment.} We now take
$t_1:=m+1=n-1$. Thus, we will peep the value of $Y_{n-1}$.
Since by 3), $\P(Y_{n-1}=3|X^n=x^n)=0$, the restricted Viterbi alignment
will differ from the original one only if $Y_{n-1}=2$. Let us find
the restricted Viterbi alignment given it passes state $2$ at time
$n-1$, i.e. let us find
\begin{align*}
v^{(1)}(x^n,2)&=\arg\max_{s^n}\P(Y^n=s^n|X^n=x^n,Y_{n-1}=2)\\
%&=\arg\max_{s^n:s_{n-1}=2}\P(Y^n=s^n|X^n=x^n,Y_{n-1}=2)\\
&=\arg\max_{s^n:s_{n-1}=2}\P(Y^n=s^n|X^n=x^n).
\end{align*}
Because of condition 2) it follows that for any $k>2$,
\begin{align*}
&\P(Y_1=1,Y_2=\cdots=Y_{n-1}=2,Y_n=1|X^n=x^n)\\
&>\P(Y_1=\cdots =Y_{k-1}=1,Y_k=\cdots=Y_{n-1}=2,Y_n=1|X^n=x^n).
\end{align*}
Secondly, since the only way from state 3 to state 2 is through
state 1, the restricted Viterbi path never visits state 3.
Therefore,  $v^{(1)}(x^n,2)$ is constantly in state 2
except the times 1 and $n$, where it equals to 1.
Thus, if $Y_{n-1}=2$, then the Viterbi and restricted Viterbi
path differ at every time from $2$ to $n-1$:
the Viterbi stays in 1 and the restricted Viterbi stays in 2.
%------------
\paragraph{Checking (\ref{badpeep}).} Since given our data,
$Y_{t_1}$  can take on two values only, we have for every
$t=1,\ldots,n$,
\begin{align*}
&\P(Y_t=v^{(1)}_t(x^n,Y_{t_1})|X^n=x^n)=
\sum_{s=1}^2\P(Y_t=v^{(1)}_t(x^n,s)|X^n=x^n,Y_{t_1}=s)\P(Y_{t_1}=s|X^n=x^n).\end{align*}
On the other hand, obviously
$$\P(Y_t=v_t(x^n)|X^n=x^n)=\sum_{s=1}^2\P(Y_t=v_t(x^n)|X^n=x^n,Y_{t_1}=s)\P(Y_{t_1}=s|X^n=x^n).$$
Because $v^{(1)}(x^n,1)=v(x^n)$ and $\P(Y_{t_1}=2|X^n=x^n)>0$, it
immediately follows that inequality (\ref{badpeep}) holds if and
only if
\begin{equation}\label{badpeep2}
\sum_{t=1}^{n}\P(Y_t=v_t(x^n)|X^n=x^n,Y_{t_1}=2)>\sum_{t=1}^{n}\P(Y_t=v^{(1)}_t(x^n,2)|X^n=x^n,Y_{t_1}=2).
\end{equation}
%------------
Recall that $t_1=n-1=m+1$. Let for every $i=1,2,3$,
$$Q_t(i):=\P(Y_t=i|X^n=x^n,Y_{n-1}=2),\quad t=1,\ldots,n.$$
With this notation, (\ref{badpeep2}) holds if and only if
\begin{equation}\label{non-neg}
 \sum_{t=2}^{m}Q_t(1)>1+\sum_{t=2}^{m}Q_t(2).
\end{equation}
%------------
%
This is indeed so in our example. Let $\delta=1$, consider
$\epsilon=0.2$ and $\epsilon=0.01$. In Table 9, the
values of the right-hand side and left-hand side of inequality
$(\ref{non-neg})$ have been calculated for some values of $m$. Observe that
for already $m=7$, inequality (\ref{non-neg}) holds. The
difference $\sum_{t=2}^m Q_t(1)-\sum_{t=2}^m Q_t(2)$ grows with
increasing $m$, and we will show that it can be made arbitrarily
large.
\newline \newline
{{\bf{Table 9}}. Comparison of accuracy before and after peeping}.
\begin{center}
\begin{tabular}{|rcc|rcc|}
  \hline
  $\epsilon=0.2$ & & & $\epsilon=0.01$ & & \\
  $m$ & $\sum_{t=2}^m Q_t(1)$ & $\sum_{t=2}^m Q_t(2)+1$  & $m$ & $\sum_{t=2}^m Q_t(1)$ & $\sum_{t=2}^m Q_t(2)+1$ \\
  \hline
  3   & 0.79 & 2.09  &  3  & 0.77  &  2.14 \\
  5   & 1.80 & 2.46  &  5  & 1.82  &  2.66 \\
  6   & 2.29 & 2.58  &  6  & 2.38  &  2.79 \\
  7   & 2.78 & 2.70  &  7  & 2.96  &  2.87 \\
  98  & 45.26 & 14.82  &  98 & 56.52  & 4.00  \\
  998 & 465.26 & 134.82 & 998 & 586.13 & 14.38 \\
  \hline
\end{tabular}
\end{center}
%------------
\paragraph{The difference $\sum_{t=2}^{m}Q_t(1) -
\sum_{t=2}^{m}Q_t(2)$ goes to infinity with $m$.}
%---------------------
At first we will show that the probabilities $Q_t(i)$ can be calculated recursively.
Let $\alpha_t(i)$ and $\beta_t(j)$ denote the usual forward and backward probabilities, i.e.
\[ \alpha_t(i)=p(x^t,Y_t=i), \quad \beta_t(j)=p(x_{t+1}^n|Y_t=j).\]
Let
\[ \gamma_{t_1,t_2}(i,j):=p(x_{t_1+1}^{t_2},Y_{t_2}=j|Y_{t_1}=i) .\]
Then for $t=2,\ldots,n-2$, $Q_t(i)$
%$Q_t(i)=P(Y_t=i|Y_{n-1}=2,X^n=x^n)$
can be expressed as
\[ Q_t(i) =\frac{\alpha_t(i)\gamma_{t,n-1}(i,2)\beta_{n-1}(2)}{\sum_i \alpha_t(i)\gamma_{t,n-1}(i,2)\beta_{n-1}(2)}=
\frac{\alpha_t(i)\gamma_{t,n-1}(i,2)}{\sum_i \alpha_t(i)\gamma_{t,n-1}(i,2)}.\]
Observe that $Q_{n-1}(1)=0$ and $Q_{n-1}(2)=Q_1(1)=Q_n(1)=1$. The quantities $\gamma_{t_1,t_2}(i,j)$ can be seen as restricted
backward probabilities. Because $x_3=x_4=\ldots=x_m=z$ and $f_1(z)=f_2(z)=f_3(z)=1$, we can calculate the forward and
restricted backward probabilities recursively. Let $u:={2\over 3}(1-\e)$ and $v:={2\over 3}\e$.
Let for any $t$,
\[ \alpha_t:=(\alpha_t(1),\alpha_t(2),\alpha_t(3))',\]
and for any $t<n-1$,
\[ \gamma_{t,n-1}:=(\gamma_{t,n-1}(1,2),\gamma_{t,n-1}(2,2),\gamma_{t,n-1}(3,2))'.\]
Then the $\alpha$-recursion is given as follows:
\[ \alpha_2(1)=\pi_1 u \, , \quad \alpha_2(2) =\pi_1 v (1+\delta)\, , \quad \alpha_2(3)={\pi_1\over 3},  \]
and for any $t=3,\ldots, m$,
\[ \alpha_t'=\alpha'_{t-1}\mathbb{P}, \quad \mbox{thus} \quad \alpha_t'=\alpha'_2 \mathbb{P}^{t-2} \, . \]
The recursion for the $\gamma$-probabilities is given as follows:
\[ \gamma_{n-2,n-1}(1,2)=v, \quad
\gamma_{n-2,n-1}(2,2)=u, \quad \gamma_{n-2,n-1}(3,2)=0, \]
and for any $t=2,\ldots,n-3$,
\[ \gamma_{t,n-1} = \mathbb{P}^{n-2-t} \gamma_{n-2,n-1}.   \]
Therefore, for any $t=2,\ldots,m$,
\[Q_t(i) = \frac{\alpha_2'\mathbb{P}^{t-2}A_i \mathbb{P}^{n-2-t} \gamma_{n-2,n-1}}
{\sum_i \alpha_2'\mathbb{P}^{t-2}A_i \mathbb{P}^{n-2-t} \gamma_{n-2,n-1}}=
\frac{\alpha_2'\mathbb{P}^{t-2}A_i \mathbb{P}^{m-t} \gamma_{n-2,n-1}}
{ \alpha_2'\mathbb{P}^{m-2} \gamma_{n-2,n-1}} \, , \]
where $A_i$ is a matrix having all entries zero except $a_{ii}=1$.
If $m\to \infty$, then
$$\mathbb{P}^{m}\to \left(
                      \begin{array}{ccc}
                        \pi_1 & \pi_2 & \pi_3 \\
                        \pi_1 & \pi_2 & \pi_3 \\
                        \pi_1 & \pi_2 & \pi_3 \\
                      \end{array}
                    \right)=:\mathbb{P}^{\infty}.$$
Hence, if $t$ is large and $m-t$ is large as well, then
$$\alpha'_t\approx \alpha'_2 \mathbb{P}^{\infty},\quad
\gamma_{t,n-1}\approx \mathbb{P}^{\infty}\gamma_{n-2,n-1},$$ so that
$$\alpha_t(s)\approx \left(\sum_i \alpha_2(i)\right)\pi_s,\quad
\gamma_{t,n-1}(s,2)\approx \sum_i \gamma_{n-2,n-1}(i,2) \pi_i.$$ Hence, if $t$ is
far from the beginning and from the end, then
$$ Q_t(s)\approx \frac{\left(\sum_i \alpha_2(i)\right)\pi_s\left(\sum_i \gamma_{n-2,n-1}(i,2) \pi_i\right)}
{\left(\sum_i \alpha_2(i)\right)\left(\sum_i \gamma_{n-2,n-1}(i,2) \pi_i\right)}= \pi_s.$$
%
%(The probabilities are basically of this size, not only far from the beginning and
%from the end.)
%
Since $\pi_1>\pi_2$, the argument above shows that choosing $m$ big
enough, the difference $\sum_{t=2}^m Q_t(1)-\sum_{t=2}^m Q_t(2)$ can
be arbitrarily large. Hence, given that $m$ is big enough, in this example peeping has definitely a
negative effect .
%-------------
\paragraph{The limit of $\P(Y_{n-1}=2|X^n=x^n)$.}
%-----------
We just saw that as $n$ grows, the difference between the left- and right-hand side
of (\ref{badpeep2}) can get arbitrarily large. This
does not necessarily imply that the difference between the left- and right-hand side
of (\ref{badpeep}) grows with $n$, unless we can show
that $\P(Y_{n-1}=2|X^n=x^n)$ is bounded away from zero as $n$ grows.
In this example this is indeed the case, since  
$\P(Y_{n-1}=2|X^n=x^n)$ converges to a non-zero limit. Since $\alpha_{n-1}'=\alpha_{n-2}'\mathbb{P}$ and
$\alpha_{n-2} \to \big(\sum_i \alpha_2(i)\big)\pi$ as $n \to \infty$,
we have
$$\alpha_{n-1}(1)\to \big(\sum_i \alpha_2(i)\big)(\pi_1 u+ \pi_2 v +
\pi_3 {1\over 2}),\quad \alpha_{n-1}(2)\to \big(\sum_i \alpha_2(i)\big)(\pi_1 v+ \pi_2 u).$$
%----------
Therefore (because $\alpha_{n-1}(3)=0$) we obtain that
$$\P(Y_{n-1}=2|X^n=x^n)= \frac{\alpha_{n-1}(2)v}{\alpha_{n-1}(1)u+\alpha_{n-1}(2)v}$$
$$\to {(\pi_1 v+ \pi_2 u)v \over (\pi_1 u+
\pi_2 v + \pi_3 {1\over 2})u+(\pi_1 v+ \pi_2 u)v}>0.$$
%-------------------------
The limit above is 0.066667 and 0.000198
for $\epsilon=0.2$ and $\epsilon=0.01$, for example.
Hence we can conclude that in our example, the difference between the left- and right-hand side of
(\ref{badpeep}) goes to infinity as $n$ grows, implying that the
expected number of additional classification errors caused by unsuccessful 
peeping can be arbitrarily large.
%-----------
\section{Appendix}
\subsection{Proof of Proposition 2.1} Let $x^n$ and $t\in
\{2,\ldots,n-1\}$ be fixed. Recall that $S=\{1,\ldots, K\}$.  Let us
estimate $\gamma_t(s)$ for any state $s\in S$ from below and from
above. Since
\begin{equation}\label{gammat}
\gamma_t(s)=\sum_{s'}\sum_{s''}\alpha(x^{t-1},s')p_{s'
s}f_s(x_t)p_{s s''}\alpha(s'',x_{t+1}^n),
\end{equation}
we have
\begin{align*}
\gamma_t(s)&\geq p(x^{t-1})(\min_{s'} p_{s's})f_s(x_t)(\min_{s'} p_{ss'})p(x_{t+1}^n),\\
\gamma_t(s)&\leq p(x^{t-1})(\max_{s'} p_{s's})f_s(x_t)(\max_{s'}
p_{ss'})p(x_{t+1}^n).
\end{align*}
%---------------------------
Assume without loss of generality that the Viterbi alignment passes
state 1 at time point $t$, that is $v_t=1$. Let $v_{t-1}=a$ and $v_{t+1}=b$.
Then for any other state $s\ne 1$ it holds that
$$p_{a1}f_1(x_t)p_{1b}\geq p_{as}f_s(x_t)p_{sb},$$
or equivalently,
\begin{equation}\label{vorra}
f_1(x_t)\geq \left( {p_{as}\over
p_{a1}}\right)f_s(x_t)\left({p_{sb}\over p_{1b}}\right).
\end{equation}
Let $s\ne 1$ be an arbitrary state. Using the upper bound for
$\gamma_t(s)$ and the lower bound for $\gamma_t(1)$ together with
(\ref{vorra}), we get
%----------------
\begin{align*}
{\gamma_t(1)\over \gamma_t(s)}
&\geq {(\min_{s'} p_{s'1})\over (\max_{s'} p_{s's})} {p_{as}\over p_{a1}}{p_{sb}\over p_{1b}}{(\min_{s'} p_{1s'})\over (\max_{s'} p_{ss'})}\\
&\geq  {(\min_{s'} p_{s'1})\over (\max_{s'} p_{s's})}{(\min_{s'}
p_{s's})\over (\max_{s'} p_{s'1})}{(\min_{s'} p_{ss'})\over
(\max_{s'} p_{1s'})}{(\min_{s'} p_{1s'})\over (\max_{s'} p_{ss'})}
\geq \sigma_1^2\sigma_2^2.
\end{align*}
Hence, for $t\in \{2,\ldots,n-1\}$, the classification probability
has the following lower bound:
$$\P(Y_t=v_t(x^n)|X^n=x^n)=\P(Y_t=1|X^n=x^n)={\gamma_t(1)\over
\sum_s\gamma_t(s)}\geq {\sigma_1^2\sigma_2^2 \over
\sigma_1^2\sigma_2^2 +(K-1)}.$$
Consider now the cases $t=1$ and $t=n$. For $t=1$, only the states with positive initial probability
are considered. For such a state $s$, equation (\ref{gammat})
becomes
\[ \gamma_1(s)=\sum_{s''}\pi_s f_s(x_1)p_{s s''}\alpha(s'',x_{2}^n).\]
%with $\alpha(x^0,s):=1$ and $p_{s's}:=\pi_s$.
For $t=n$ and any $s\in S$,
\[\gamma_n(s)=\sum_{s'} \alpha(x^{n-1},s')p_{s's}f_s(x_n). \]
%(\ref{gammat}) holds with $\alpha(s'',x_{n+1}^n):=1$ and $p_{ss''}=:1$.
Similarly, the ratios in (\ref{vorra}) become for $t=1$ and $t=n$, respectively,
%\begin{equation}\label{vorra2}
\[f_1(x_1)\geq \left( {\pi_{s}\over
\pi_{1}}\right)f_s(x_1)\left({p_{sb}\over p_{1b}}\right),\quad
f_1(x_n)\geq \left( {p_{as}\over p_{a1}}\right)f_s(x_n). \]
%\end{equation}
Thus,
\begin{align*}
{\gamma_1(1)\over \gamma_1(s)} &\geq {\pi_1 \over \pi_s}
{\pi_{s}\over \pi_{1}} {p_{sb}\over p_{1b}}{(\min_{s'}
p_{1s'})\over (\max_{s'} p_{ss'})}\geq
{(\min_{s'}p_{ss'})\over (\max_{s'}p_{1s'})}{(\min_{s'} p_{1s'})\over (\max_{s'} p_{ss'})}\geq \sigma_1^2,\\
{\gamma_n(1)\over \gamma_n(s)}&\geq {(\min_{s'}p_{s'1})\over
(\max_{s'}p_{s's})}{p_{as}\over p_{a1}}\geq {(\min_{s'}p_{s'1})\over
(\max_{s'}p_{s's})}{(\min_{s'}p_{s's})\over (\max_{s'}p_{s'1})}\geq
\sigma_2^2,
\end{align*}
and the corresponding bounds for the classification probabilites are
$$\P(Y_1=v_1(x^n)|X^n=x^n)\geq {\sigma_1^2 \over
\sigma_1^2 +(K_1-1)},\quad \P(Y_n=v_n(x^n)|X^n=x^n)\geq {\sigma_2^2
\over \sigma_2^2 +(K-1)}.$$
%
%
%
%-----------------------
\subsection{Proof of Proposition 2.2}
%---------------
To prove the proposition, we use Lemma 5.1 from \cite{doob}. We present the lemma using the same notation as in \cite{doob}.
The random variables of the Markov chain are denoted by $x_1,x_2,\ldots$, the state space is denoted by $X$, and $\mathcal F_X$
is a Borel field of $X$ sets. Let for $A \in \mathcal F_X$, $p(\xi,A)=\sum_{\eta \in A} p_{\xi \eta}$, and let
$p^{(n)}(\xi,A)$ denote the corresponding $n$-step probability.
 The conditional probability that (from initial point $\xi$) the system will be in a state of $A \in \mathcal F_X$
at some time during the first $n$ transitions, is denoted by $\tilde{p}^{(n)}(\xi,A)$, that is
\[ \tilde{p}^{(n)}(\xi,A)=\P\{\cup_{j=2}^{n+1} [x_j(\omega)\in A]|x_1(\omega)=\xi \}. \]
Hypothesis (D) in \cite{doob} is the Doeblin condition.
\newline \newline
{\bf{Hypothesis (D)}} {\textit{There is a (finite-valued) measure $\varphi$ of sets $A \in \mathcal F_X$ with $\varphi(X)>0$,
an integer $\nu \geq 1$, and a positive $\varepsilon$, such that
\[ p^{(\nu)}(\xi,A) \leq 1-\varepsilon \quad \mbox{if} \quad \varphi(A) \leq \varepsilon \, .\] }}
Hypothesis (D) is always satisfied in the case of finite state space, thus it imposes no restriction on
finite dimensional stochastic matrices.
%-----------
%
\begin{lemma}\label{5.1} (\textrm{Doob, 1953}) Under Hypothesis (D), if a set $A \in \mathcal F_X$
has the property that
\begin{equation} \label{doobprop}
\lim_{n \to \infty} \tilde{p}^{(n)}(\xi,A)=\sup_n \tilde{p}^{(n)}(\xi,A)>0, \quad \forall \xi \in X,
\end{equation}
then there is a positive integer $\mu$ and a positive $\rho < 1$ for
which
\[ \tilde{p}^{(n)}(\xi,A) \ge 1-\rho^{(n/\mu)-1}, \quad \xi \in X. \]
\end{lemma}
Lemma \ref{5.1} is proved by induction. \\\\
%------------
Recall that $W^{\ast}_t=\min\{w>t+r: \,\, X_{w-r}^w \in {\mathcal X}_o^{r+1}\}$.
Consider an arbitrary $t$. Then $\P(W^{\ast}_t-t>k)= 1 -\P(W^{\ast}_t \le t+k).$
Suppose $W^{\ast}_t=t+l$ for some $l>r$. Then $X_{t+l-r}^{t+l}\in
{\mathcal X}_o^{r+1}$. Since $\forall x \in {\mathcal X}_o$,
$\max_{s \notin C} f_s(x)=0$, we are interested in only those state
paths, where $Y_{t+l-r}^{t+l}\in C^{r+1}$. To prove the proposition, we define two new Markov chains
$U$ and $Z$, and consider an equivalent event to $\{W^{\ast}_t \le t+k\}$ for the chain $Z$. To $Z$, we can apply
Doob's lemma. \\\\
%---------------
We start with defining a new Markov chain $U=\{U_t\}_{t=1}^{\infty}:=\{Y_t,I_{X_t}(\mathcal X_o)\}_{t=1}^{\infty}$,
where
\[
I_{X_t}(\mathcal X_o)=\left\{
               \begin{array}{ll}
                 1, & \hbox{if $X_t \in \mathcal X_o$;} \\
                 0, & \hbox{if $X_t \notin \mathcal X_o$.}
               \end{array}
             \right.
\]
Since $\forall x \in {\mathcal X}_o$, $f_s(x)=0$ when $s \notin C$,
the states where $Y_t \notin C$ and $X_t \in \mathcal X_o$ are not
possible. Thus, the state space of $U$ has $K+|C|=S_U$ possible states.
The transition probabilities for $U$ are given by a matrix
$\mathcal P$ as follows: let $u_t=(i,k)$ and $u_{t+1}=(j,l)$, then
\[ \mathcal P(u_t,u_{t+1})=\P(U_{t+1}=u_{t+1}|U_t=u_t)=
    \P(Y_{t+1}=j,I_{X_{t+1}}(\mathcal X_o)=l |Y_{t}=i, I_{X_t}(\mathcal X_o)=k)  \]
\[ =\P( I_{X_{t+1}}(\mathcal X_o)=l|Y_{t+1}=j)\P(Y_{t+1}=j|Y_t=i)=
\left\{
               \begin{array}{ll}
                p_{ij} P_j(\mathcal X_o), & \hbox{if $l = 1$;} \\
                p_{ij} P_j(\mathcal X_o^c), & \hbox{if $l = 0$.}
               \end{array}
             \right. \]
Observe that if $j \notin C$, then $P_j(\mathcal X_o^c)=1$. Define
now the Markov chain $Z=\{Z_t\}_{t=1}^{\infty}$ as
\[ Z_t:=(U_t,U_{t+1},\ldots,U_{t+r}).\]
This chain has $S_U^{r+1}$ possible states and the transitions for
$Z$ are determined by the transition probabilities for $U$. A
transition from $Z_t$ to $Z_{t+1}$ is possible only if the last $r$
elements of $Z_t$ and the first $r$ elements of $Z_{t+1}$ coincide.
The transition probability in this case is given by $\mathcal
P(u_{t+r},u_{t+r+1})$.

Let $H$ denote the subset of states of $Z$, such that for all $U_j$
in $Z_t$, $j=t,\ldots,t+r$, $Y_j \in C$ and $I_{X_j}(\mathcal X_o)=1$,
i.e.~$Y_t^{t+r} \in C^{r+1}$ and $X_t^{t+r}\in \mathcal X_o^{r+1}$.
There are $|C|^{r+1}$ such possible states. Then the event
$\left\{\bigcup_{i=t+1}^{t+k-r}(Z_i \in H)\right\}$ is equivalent to the event
$\{W_t^{\ast} \le t+k \}$.\\\\
To apply Doob's lemma, we have to check that property (\ref{doobprop}) holds for $Z$
and our set $H$.
We have for large $n$:
\[ \P\left\{\cup_{j=2}^{n+1}(Z_j \in H)|Z_1=\xi \right\} \geq \P(Z_n\in H |Z_1 = \xi) \]
\[  = \P(Y_n^{n+r}\in C^{r+1}, X_n^{n+r} \in \mathcal X_o^{r+1}|Z_1=\xi)\]
\[ = \P( X_n^{n+r}\in \mathcal X_o^{r+1}|Y_n^{n+r}\in C^{r+1})\P(Y_n^{n+r} \in C^{r+1} |Z_1=\xi) =:prob_1 \cdot prob_2. \]
%----------------------
Consider at first $prob_1$. According to the cluster definition,
$\min_{s \in C} f_s(x) > \epsilon$ for some $\epsilon >0$
for every $x \in \mathcal X_o$. Therefore, $\int_{\mathcal X_o} f_s(x) d\mu >
\epsilon \mu(\mathcal X_o)=m$ if $s \in C$. Thus,
\[ prob_1= \prod_{t=n}^{n+r} \P(X_t\in \mathcal X_o|Y_t\in C) > m^{r+1}. \]
%-----------------
Consider now $prob_2$. Recall that $R=(p_{ij})_{i,j\in C}$ and due to {\bf A1}, $R^r$ is
strictly positive. Therefore, $\min_{i \in C} p^{(r)} (i,C)=\min_{i
\in C} \sum_{j \in C} p^{(r)}_{ij}>\delta$ for some $\delta >0$. Let the state of $Y_{r+1}$ in
$Z_1=\xi$ be $s$. We obtain:
\[ prob_2=\P(Y_n \in C,Y_{n+1} \in C,\ldots,Y_{n+r} \in C,Z_1=\xi)/\P(Z_1=\xi)\]
\[ =\sum_{i\in C}\sum_{j\in C} \P(Y_{n+1} \in C,\ldots,Y_{n+r-1} \in C,Y_{n+r}=j|Y_n=i,Z_1=\xi)\P(Y_n=i|Z_1=\xi) \]
\[ =\sum_{i \in C} p^{(r)} (i,C) \P(Y_n=i|Z_1=\xi)> \delta \sum_{i \in C} \P(Y_n=i|Y_{r+1}=s)=
\delta \sum_{i \in C} p^{(n-r-1)}_{si}.  \]
Since $Y$ is irreducible, there exist $n_s$ and $\eta_s>0$ for every $s\in S$
such that $\sum_{i \in C} p^{(n_s-r-1)}_{si}=\eta_s >0$. Take
$\eta^{\ast}=\min_s \eta_s$ and $n^{\ast}=\max_s n_s$. Then since
$\P\left\{\cup_{j=2}^{n+1} (Z_j \in H)|Z_1=\xi \right\}$ is monotone
and nondecreasing, we have that for $n>n^{\ast}$,
\[ \P\left\{\cup_{j=2}^{n+1} (Z_j \in H)|Z_1=\xi \right\}>m^{r+1} \delta \eta^{\ast}.\]
Observe that this holds for every $t$, i.e.~when we condition on $Z_t$ and take the union over $\{t+1,\ldots,t+n\}$.
Now we can prove Proposition 2.2.
\paragraph{Proof of Proposition 2.2.} We have:
\[ \P(W^{\ast}_t-t>k)= 1 -\P(W^{\ast}_t \le t+k)=1-\P\left\{\cup_{i=t+1}^{t+k-r} (Z_i \in H)\right\}\]
\[  = 1-\sum_{\xi}\P\left\{\cup_{i=t+1}^{t+k-r} (Z_i \in H)|Z_t=\xi \right\} \P(Z_t =\xi) \stackrel{(\text{Lemma \ref{5.1})}}{\le}
 \rho^{(k-r)/\mu-1} = a \exp[-bk],\]
where $a=\rho^{-r/\mu-1}$ and $b=-{1 \over \mu}\ln{\rho}$.
%-------------------------------------
\subsection{Emission matrix for Subsection 3.2.1}
$$ \left( \begin{array}{cccccc}
       P_1 &  P_2 &   P_3 &  P_4 & P_5  &  P_6   \\
 0.1059 & 0.0636 & 0.0643 & 0.1036 & 0.1230 & 0.1230 \\
 0.0107 & 0.0171 & 0.0135 & 0.0081 & 0.0111 & 0.0128 \\
 0.0538 & 0.0319 & 0.0775 & 0.0634 & 0.0415 & 0.0345 \\
 0.0973 & 0.0477 & 0.0620 & 0.1120 & 0.0852 & 0.0848 \\
 0.0436 & 0.0576 & 0.0330 & 0.0371 & 0.0386 & 0.0399 \\
 0.0303 & 0.0484 & 0.1133 & 0.0447 & 0.0321 & 0.0229 \\
 0.0203 & 0.0227 & 0.0259 & 0.0188 & 0.0197 & 0.0221 \\
 0.0564 & 0.1010 & 0.0372 & 0.0577 & 0.0694 & 0.0593 \\
 0.0672 & 0.0443 & 0.0574 & 0.0540 & 0.0671 & 0.0810 \\
 0.1227 & 0.1068 & 0.0674 & 0.0994 & 0.1279 & 0.1477 \\
 0.0240 & 0.0219 & 0.0181 & 0.0214 & 0.0293 & 0.0304 \\
 0.0299 & 0.0252 & 0.0561 & 0.0259 & 0.0338 & 0.0336 \\
 0.0333 & 0.0208 & 0.0757 & 0.0472 & 0.0067 & 0.0031 \\
 0.0443 & 0.0270 & 0.0330 & 0.0469 & 0.0497 & 0.0472 \\
 0.0594 & 0.0464 & 0.0470 & 0.0522 & 0.0677 & 0.0697 \\
 0.0496 & 0.0496 & 0.0744 & 0.0485 & 0.0422 & 0.0491 \\
 0.0395 & 0.0641 & 0.0572 & 0.0465 & 0.0412 & 0.0375 \\
 0.0591 & 0.1386 & 0.0473 & 0.0685 & 0.0677 & 0.0545 \\
 0.0168 & 0.0170 & 0.0111 & 0.0135 & 0.0130 & 0.0124 \\
 0.0359 & 0.0483 & 0.0286 & 0.0306 & 0.0331 & 0.0345 \\
\end{array} \right)$$
\bibliographystyle{plain}
\bibliography{peep}
\end{document}